\documentclass[11pt]{article}
\usepackage{amsmath,amssymb,amstext,bm, graphicx, amsthm}
\usepackage{graphicx}
\usepackage{epsf,pstricks, xcolor, verbatim}
\newtheorem{theorem}{Theorem}
\newtheorem{lemma}[theorem]{Lemma}
\newtheorem{corollary}[theorem]{Corollary}
\newtheorem{proposition}[theorem]{Proposition}
\newtheorem{fact}[theorem]{Fact}
\newtheorem{claim}[theorem]{Claim}
\newtheorem{definition}[theorem]{Definition}
\newtheorem{remark}[theorem]{Remark}

\def\appendices{\par
    \setcounter{section}{0}\setcounter{subsection}{0}
    \def\thesection{Appendix \Roman{section}}
}

\newcommand{\an}{{n_{00}}}

\renewcommand{\S}{\CalS}

\newcommand{\badk}{{k^*}}

\newcommand{\pr}{\prime}

\newcommand{\E}{{\bf E}}

\newcommand{\al}{n}
\newcommand{\sig}{\sigma}
\newcommand{\eps}{\epsilon}

\newcommand{\R}{\mathbb R}

\newcommand{\C}{\mathbb C}
\newcommand{\F}{\mathbb F}

\newcommand{\CalCN}{\mathcal{CN}}
\newcommand{\CalN}{\mathcal{N}}
\newcommand{\CalC}{{\mathcal{C}}}

\newcommand{\CalS}{{\mathcal{S}}}

\renewcommand{\b}{\text{solid}}
\renewcommand{\r}{\text{dashed}}
\newcommand{\snr}{\mathtt{SNR}}
\newcommand{\inr}{\mathtt{INR}}
\newcommand{\n}{N_0}
\newcommand{\CN}{\mathcal{CN}}
\newcommand{\underC}{\underline{\CalC}}
\newcommand{\overC}{\overline{\CalC}}
\newcommand{\nt}{{\tilde n}}

\renewcommand{\phi}{\varphi}
\newcommand{\lf}{\lfloor}
\newcommand{\rf}{\rfloor}
\newcommand{\y}{\tilde{y}}

\begin{document}

\title{The Approximate Capacity of the Many-to-One and One-to-Many Gaussian Interference Channels}
\author{ Guy Bresler \and
Abhay Parekh \and David N. C. Tse \thanks{This work appeared without full proofs in an extended abstract presented at the Allerton Conference on Communication, Control, and Computing \cite{ManyToOneAllerton}, September, 2007. This work was supported by a Vodafone-US Foundation Graduate Fellowship, an NSF Graduate Research Fellowship, and by the National Science Foundation under an ITR grant: the 3 Rs of Spectrum Management: Reuse, Reduce and Recycle.} \thanks{The authors are with Wireless Foundations, Department of EECS, UC Berkeley, Berkeley, California, USA. Email: {\sffamily \{gbresler, parekh, dtse\}@eecs.berkeley.edu} }
}

\maketitle

\begin{abstract}
Recently, Etkin, Tse, and Wang found the capacity region of the two-user Gaussian interference channel to within one bit/s/Hz. A natural goal is to apply this approach to the Gaussian interference channel with an arbitrary number of users. 
We make progress towards this goal by finding the capacity region of the many-to-one and one-to-many Gaussian interference channels to within a constant number of bits. The result makes use of a deterministic model to provide insight into the Gaussian channel. The deterministic model makes explicit the dimension of signal scale. A central theme emerges: the use of lattice codes for alignment of interfering signals on the signal scale. 
\end{abstract}


\section{Introduction}
Finding the capacity region of the two user Gaussian interference channel is a long-standing open problem. 
Recently, Etkin, Tse, and Wang \cite{OneBit} made progress on this problem by finding the capacity
region to within one
bit/s/Hz. In light of the difficulty in finding the exact capacity regions of most Gaussian channels, their result introduces a fresh approach towards understanding multiuser Gaussian channels.
A natural goal is to apply their approach to the Gaussian
interference channel with an arbitrary number of users. This paper
makes progress towards this goal by considering two special cases---the
many-to-one and one-to-many interference channels (IC)---where interference is experienced, or is caused, by only one user. The capacity regions of the many-to-one and one-to-many Gaussian ICs are
determined to within a constant gap, independent of the channel
gains. For the many-to-one IC, the size of the gap is less than $(2 K+5) \log K$ bits per user,
where $K$ is the number of users. For the one-to-many IC, the gap is $2K+1$ bits for user 0 and 1 bit for each of the other users.
This result establishes, as a
byproduct, the generalized degrees-of-freedom regions of these channels,
as defined in \cite{OneBit}.

Despite interference occurring only at one user, the capacity regions of the many-to-one and one-to-many ICs
exhibit an interesting combinatorial structure, and new outer
bounds are required. To elucidate this structure, we make use of a
particular deterministic channel model, first introduced in
\cite{DeterministicRelay}; this model retains the essential features
of the Gaussian channel, yet is significantly simpler. We show that
the capacity regions of the deterministic and Gaussian channels are
closely related to one another, and in fact, the generalized degrees of freedom
region of the Gaussian channel is \emph{equal} to the capacity
region of an appropriate deterministic channel.

While the derivation of the outer bound for the many-to-one Gaussian IC parallels that
of the deterministic case, the achievable strategy for the Gaussian
channel is noteworthy. In order to successfully emulate the strategy
for the deterministic channel in the Gaussian setting, it is necessary
to use lattice codes. The idea is that since there are multiple
interferers, they should align their interference so as to localize
the aggregate effect; the impact of the interference is practically
as though from one user only.  The idea of interference alignment was introduced in a different setting for the MIMO X channel by Maddah-Ali, \emph{et. al.} \cite{MMK08} and for the many-user interference channel by Cadambe and Jafar \cite{Jafar}. In those works alignment is achieved in the signal space; in this paper, alignment is achieved on the signal \emph{scale}. 
Lattice codes, rather than
random codes, are used to achieve this localization. In Section~\ref{s:ExampleManytoOne} we consider an example using a simple generalization (to many users) of the Han-Kobayashi scheme with Gaussian codebooks. We show that this random coding strategy cannot achieve the degrees-of-freedom of the many-to-one Gaussian IC.

Lattice strategies are a natural solution to certain multiuser
problems, and several examples have recently been found for which
lattice strategies achieve strictly better performance than any
known random codes, including: the work of Nazer and Gastpar on computation over multiple access
channels \cite{Bobak}, and Philosoph, \emph{et. al.}'s dirty paper coding for multiple access
channels \cite{DirtyMAC}.

In contrast to the many-to-one IC, the one-to-many IC is simpler, requiring only Gaussian random codebooks. In particular, a generalized Han-Kobayashi scheme with Gaussian random codebooks is essentially optimal. As in the many-to-one IC, a deterministic channel model guides the development. Moreover, the deterministic channel model reveals the relationship between the two channels:
the capacity regions of 
the deterministic many-to-one IC and one-to-many IC, obtained by reversing the role of transmitters and receivers, are identical, i.e. the channels are reciprocal. This relationship is veiled in the Gaussian setting, where the statement holds only in an approximate sense. 

While the many-to-one IC is more theoretically interesting, requiring a new achievable scheme using lattices to align interference, 
the one-to-many IC seems more practically relevant. One easily imagines a scenario with one powerful long-range transmit-receive link and many weak short-range links sharing the wireless medium. Here, to a good approximation, there is no interference except from the single powerful transmitter.

Using the deterministic model and the framework developed in this work, Cadambe et al. \cite{CJS08} found a sequence of symmetric $K$-user Gaussian interference channels with arbitrarily close to $K/2$ total degrees of freedom, and
Jafar and Vishwanath \cite{JV08} show that the generalized degrees-of-freedom region of the fully symmetric many-user interference channel (with all the signal-to-noise ratios equal to $\snr$ and all interference-to-noise ratios equal to $\inr=\snr^\alpha$) is independent of the number of users and is identical to the 2-user case except for a singularity at $\alpha=1$ where the degrees of freedom per user is $\frac1K$.

Independently, Jovicic, Wang, and Viswanath \cite{ITW_JWV} considered the many-to-one and one-to-many interference channels. They found the capacity to within a constant gap for the special case where the direct gains are greater than the cross gains. 
In this case, Gaussian random codebooks and pairwise constraints for the outer bound are sufficient. As mentioned above, for the many-to-one IC with arbitrary gains, Gaussian random codebooks are suboptimal; also, in general for both the many-to-one and one-to-many channels a sum-rate constraint is required for each subset of users. 

The paper is organized as follows. Section~\ref{s:MotivatingExample} introduces the Gaussian many-to-one IC and studies a simple example channel that motivates the entire paper. 
Section~\ref{s:DetModel} presents the deterministic channel model. Then, in Section~\ref{s:DeterministicManytoOneCapacity}, the
capacity region of the deterministic many-to-one IC is established. Section~\ref{s:GaussianManytoOne} focuses on the Gaussian many-to-one IC and
finds the capacity to within a constant gap. Finally, sections~\ref{s:OneToManyDet} and~\ref{s:OneToManyGaussian} consider the one-to-many interference channel, show that the corresponding deterministic model is reciprocal to the many-to-one channel, and approximate the capacity of the Gaussian channel to within a constant gap.

\section{Gaussian Interference Channel and Motivating Example}\label{s:MotivatingExample}
\subsection{Gaussian Interference Channel Model}
We first introduce the  multi-user Gaussian interference channel. For notational simplicity in the sequel, we assume there are $K+1$ users, labeled $0,1,\dots,K$. The channel outputs are related to the inputs by
\begin{equation}y_i = \sum_{j = 0} ^K h_{ij}x_j+z_i,\quad 0\leq i\leq K \end{equation} where for $0\leq i\leq K$, $x_i\in \C$ is subject to a unit average power constraint $\frac{1}{N}E||x_i^N||^2\leq P_i$ and the noise processes $z_i\sim \CalCN(0,N_0)$ are i.i.d. over time. 
The channel gain between input $i$ and output $j$ is denoted by $h_{ji}\in\C$. 
The signal-to-noise and interference-to-noise ratios are defined as $\snr_i=|h_{ii}|^2 P_i/N_0$ for $0\leq i\leq N$, and $\inr_{ij}=|h_{ji}|^2 P_i/N_0$ for $0\leq i,j\leq N$, $i\neq j$. Each receiver attempts to decode the message from its corresponding receiver. The remaining (standard) definitions can be found in \cite{OneBit}, and generalize naturally to an arbitrary number users.

\begin{figure}
\begin{centering}
\psset{unit=.8mm,arrowlength=1,arrowinset=0}
\begin{center}
\begin{pspicture}(-16,-32)(20,44)
    \rput(-16,35){%
            \psline(0,0)(2,3)(-2,3)(0,0)
            \psline(0,0)(0,-3)(-1.4,-3)
            \psline{->}(2.5,1)(28.5,1)
            \pscircle(30,1){1.5}
            \psline{-}(30,.25)(30,1.75) \psline{-}(29.25,1)(30.75,1)
            \psline{->}(31.5,1)(37,1)
            \pscircle(39,1){2}
            \psline{-}(39,-1)(39,-4)
            \uput[d](0,-3){$\text{Tx}_0$}
            \uput[l](-2,0){$x_0$}
            \uput[d](39,-3){$\text{Rx}_0$}
            \uput{2pt}[u](12,1){\small $h_{00}$}
            \uput{3pt}[u](34,1){$y_0$}
            \psline{->}(30,6.25)(30,2.5)
            \uput[u](30,5){$z_0$}
            }

    \rput(-16,15){%
            \psline(0,0)(2,3)(-2,3)(0,0)
            \psline(0,0)(0,-3)(-1.4,-3)
            \psline{->}(2.5,1)(28.5,1)
            \pscircle(30,1){1.5}
            \psline{-}(30,.25)(30,1.75) \psline{-}(29.25,1)(30.75,1)
            \psline{->}(31.5,1)(37,1)
            \pscircle(39,1){2}
            \psline{-}(39,-1)(39,-4)
            \uput[d](0,-3){$\text{Tx}_1$}
            \uput[l](-2,0){$x_1$}
            \uput[d](39,-3){$\text{Rx}_1$}
            \uput{2pt}[u](12,1){\small $h_{11}$}
            \uput{3pt}[u](34,1){$y_1$}
            \psline{->}(30,6.25)(30,2.5)
            \uput[u](30,5){$z_1$}
            }
    \psline{->}(-13.5,16)(12.5,35.6)
    \rput(-2,28.3){\small $h_{01}$}

    \rput(-16,-5){%
            \psline(0,0)(2,3)(-2,3)(0,0)
            \psline(0,0)(0,-3)(-1.4,-3)
            \psline{->}(2.5,1)(28.5,1)
            \pscircle(30,1){1.5}
            \psline{-}(30,.25)(30,1.75) \psline{-}(29.25,1)(30.75,1)
            \psline{->}(31.5,1)(37,1)
            \pscircle(39,1){2}
            \psline{-}(39,-1)(39,-4)
            \uput[d](0,-3){$\text{Tx}_2$}
            \uput[l](-2,0){$x_2$}
            \uput[d](39,-3){$\text{Rx}_2$}
            \uput{2pt}[u](22,1){\small $h_{22}$}
            \uput{3pt}[u](34,1){$y_2$}
            \psline{->}(30,6.25)(30,2.5)
            \uput[u](30,5){$z_2$}
            }
    \psline{->}(-13.5,-4)(13,35)
    \rput(-4,4.2){\small $h_{02}$}

    \rput(-16,-18){\Large{$\vdots$}}
    \rput(23,-18){\Large{$\vdots$}}

    \rput(-16,-30){%
            \psline(0,0)(2,3)(-2,3)(0,0)
            \psline(0,0)(0,-3)(-1.4,-3)
            \psline{->}(2.5,1)(28.5,1)
            \pscircle(30,1){1.5}
            \psline{-}(30,.25)(30,1.75) \psline{-}(29.25,1)(30.75,1)
            \psline{->}(31.5,1)(37,1)
            \pscircle(39,1){2}
            \psline{-}(39,-1)(39,-4)
            \uput[d](0,-3){$\text{Tx}_K$}
            \uput[l](-2,0){$x_K$}
            \uput[d](39,-3){$\text{Rx}_K$}
            \uput{2pt}[u](22,1){\small $h_{KK}$}
            \uput{3pt}[u](34,1){$y_K$}
            \psline{->}(30,6.25)(30,2.5)
            \uput[u](30,5){$z_K$}
            }
    \psline{->}(-13.5,-29)(13.8,34.7)
    \rput(-4,-18){\small $h_{0K}$}

\end{pspicture}
\end{center}
\caption{The Gaussian many-to-one IC: $K$ users all causing interference at receiver 0.}
\label{ZChannelFigure}
\end{centering}
\end{figure}
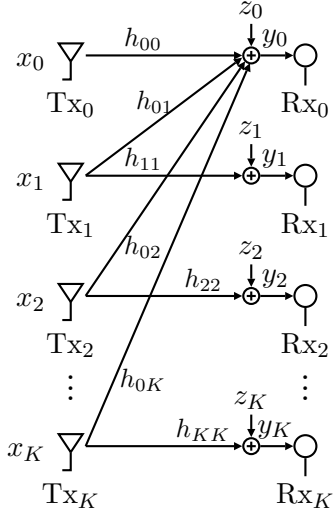

In this paper we consider two special cases of the Gaussian IC. In the first half we study the Gaussian \emph{many-to-one} IC, where all gains are zero
except $h_ {ii} $, $0\leq i\leq K$, and $h_ {i0} $, $ 1\leq i\leq K
$. The channel is depicted in Figure~\ref{ZChannelFigure}. In the second half of the paper we treat the Gaussian \emph{one-to-many} IC, which is obtained from the many-to-one IC by reversing the roles of transmitters and receivers. The one-to-many IC has all gains equal to zero except for $h_ {ii} $, $0\leq i\leq K$, and $h_ {0i} $, $ 1\leq i\leq K
$.

\subsection{Motivating Example}

\label{s:ExampleManytoOne}
In the two-user Gaussian IC a simple Han-Kobayashi scheme with Gaussian codebooks was shown to be nearly optimal \cite{OneBit}. A natural question: is the same type of scheme, a (generalized) Han and Kobayashi scheme with Gaussian codebooks, nearly optimal with more than 2 users?
We answer this question by way of an example 3-user Gaussian many-to-one channel. This example goes to the heart of the problem and captures the salient features of the many-to-one channel. In particular, the approach used for the two-user interference channel is demonstrated to be inadequate for three or more users, while a simple strategy that aligns interference on the signal scale is shown to be asymptotically optimal.  

The example channel is depicted in Figure~\ref{fig:HKexample}. The power constraints are $P_0=P_1=P_2=1$ and the gains are $h_{00}=h_{01}=h_{02}=\beta$ and $h_{11}=h_{22}=\sqrt{\beta}$. 

We first describe the Han-Kobayashi scheme with Gaussian codebooks for the three-to-one channel. In the many-to-one channel each user's signal causes interference only at receiver 0, so the signals from users 1 and 2 are each split into common and private parts as in the two-user scheme (see \cite{OneBit}, \cite{HanKobayashi}, for details on the two-user scheme). 
Each user $i=1,2$ employs a superposition of Gaussian codebooks
$$X_i=U_i+W_i\,,$$ with power $P_{U_i}+P_{W_i}=1$ and rates $R_{U_i},R_{W_i}$. The ``private" signal $U_i$ is treated as noise by receiver 0, while the ``common" signal $W_i$ is decoded by receiver 0. User 0 selects a single random Gaussian codebook with power $P_0$ and rate $R_0$.

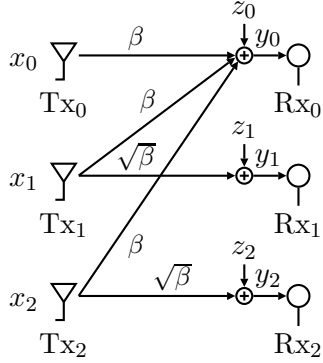
\begin{figure}
\centering
\psset{unit=0.8mm,arrowlength=1,arrowinset=0}
\begin{center}
\begin{pspicture}(-18,-10)(20,44)
    \rput(-16,35){%
            \psline(0,0)(2,3)(-2,3)(0,0)
            \psline(0,0)(0,-3)(-1.4,-3)
            \psline{->}(2.5,1)(28.5,1)
            \pscircle(30,1){1.5}
            \psline{-}(30,.25)(30,1.75) \psline{-}(29.25,1)(30.75,1)
            \psline{->}(31.5,1)(37,1)
            \pscircle(39,1){2}
            \psline{-}(39,-1)(39,-4)
            \uput[d](0,-3){$\text{Tx}_0$}
            \uput[l](-2,0){$x_0$}
            \uput[d](39,-3){$\text{Rx}_0$}
            \uput{2pt}[u](12,1){\small $\beta$}
            \uput{3pt}[u](34,1){$y_0$}
            \psline{->}(30,6.25)(30,2.5)
            \uput[u](30,5){$z_0$}
            }

    \rput(-16,15){%
            \psline(0,0)(2,3)(-2,3)(0,0)
            \psline(0,0)(0,-3)(-1.4,-3)
            \psline{->}(2.5,1)(28.5,1)
            \pscircle(30,1){1.5}
            \psline{-}(30,.25)(30,1.75) \psline{-}(29.25,1)(30.75,1)
            \psline{->}(31.5,1)(37,1)
            \pscircle(39,1){2}
            \psline{-}(39,-1)(39,-4)
            \uput[d](0,-3){$\text{Tx}_1$}
            \uput[l](-2,0){$x_1$}
            \uput[d](39,-3){$\text{Rx}_1$}
            \uput{2pt}[u](12,1){\small $\sqrt \beta$}
            \uput{3pt}[u](34,1){$y_1$}
            \psline{->}(30,6.25)(30,2.5)
            \uput[u](30,5){$z_1$}
            }
    \psline{->}(-13.5,16)(12.5,35.6)
    \rput(-2,28.3){\small $\beta$}

    \rput(-16,-5){%
            \psline(0,0)(2,3)(-2,3)(0,0)
            \psline(0,0)(0,-3)(-1.4,-3)
            \psline{->}(2.5,1)(28.5,1)
            \pscircle(30,1){1.5}
            \psline{-}(30,.25)(30,1.75) \psline{-}(29.25,1)(30.75,1)
            \psline{->}(31.5,1)(37,1)
            \pscircle(39,1){2}
            \psline{-}(39,-1)(39,-4)
            \uput[d](0,-3){$\text{Tx}_2$}
            \uput[l](-2,0){$x_2$}
            \uput[d](39,-3){$\text{Rx}_2$}
            \uput{2pt}[u](18,1){\small $\sqrt \beta$}
            \uput{3pt}[u](34,1){$y_2$}
            \psline{->}(30,6.25)(30,2.5)
            \uput[u](30,5){$z_2$}
            }
    \psline{->}(-13.5,-4)(13,35)
    \rput(-4,4.2){\small $\beta$}
\end{pspicture}
\end{center}
\caption{The example Gaussian channel, with $\snr_1=\snr_2=\beta$ and $\snr_0=\inr_1=\inr_2=\beta^2$, for some $\beta>1$.}%
\label{fig:HKexample}%
\end{figure}

There are two equivalent ways of interpreting the effect of the interference to receiver 0 from users 1 and 2: the first asks how ``noisy" is the interference, while the second asks how much of the information content within the interference is available to receiver 0. Both viewpoints are related by the entropy $h(\beta x_1+\beta x_2+z_0)$. For the first, we may expand the mutual information relevant to user 0 as $I(x_0;y_0)=h(y_0)-h(\beta x_1+\beta x_2+z_0)$. If the second entropy term is large, i.e. the interference is ``noisy", then the rate of user 0 must be small.

Consider now the second viewpoint. Note that the interference to noise ratios $\inr_1=\inr_2=\beta$ are much larger than the signal-to-noise ratios $\snr_1=\snr_2=\sqrt \beta$. Let us momentarily recall  a similar situation in the context of the two-user channel: the so-called strong interference regime occurs when the cross-gains are stronger than the direct gains. Hence, after decoding the intended signal, each receiver can decode the interfering signal as well, and thus each signal consists entirely of common information. This means the interference is quite damaging, since it contains the full information content of the intended signal.

Returning to our example channel, let us examine the output at receiver 0, ignoring the intended signal from transmitter 0. The information on $x_1$ and $x_2$ at receiver 0 is then $I(y_0;x_1,x_2|x_0)=h(\beta x_1+\beta x_2+z_0)-h(z_0)$. For our example channel it turns out that when $x_1$ and $x_2$ are Gaussian distributed, the entropy $h(\beta x_1+\beta x_2+z_0)$ is large enough to allow user 0 to decode both of the signals $x_1$ and $x_2$.
Thus, the signals from users 1 and 2 are entirely common information. When all of the signals are common information, it is easy to bound the sum-rate $r_\text{sum}^{HK}=r_0+r_1+r_2$, since the rates must lie within the MAC capacity region formed by receiver 0 and the three transmitters. This reasoning yields the following claim.

\begin{claim} \label{l:HKsuboptimal}A Han and Kobayashi type scheme, with codebooks drawn from the Gaussian distribution, and each of users $1$ and $2$ splitting their signal into independent private and common information, attains a sum-rate $r_\text{sum}^{HK}$ of at most $\log(1+3\beta^2)$. That is, with this strategy 
\begin{equation}\label{e:ExampleRHKsum}
r_\text{sum}^{HK}=r_0+r_1+r_2\leq \log(1+3\beta^2)\approx 2\log \beta.
\end{equation}
\end{claim}
\begin{proof}
The argument bears some resemblance to that of Sato \cite{Sato} in his treatment of the two-user channel under strong interference.
However, here we must show that each of the private and common messages from users 1 and 2 can be decoded by receiver 0. We quickly summarize the argument. Note first that for an achievable rate point, we may assume that each of receivers 0, 1, and 2 is able to decode their intended signal. Upon decoding signal 0, receiver 0 can subtract it off. The rate tuple $(R_{U_1},R_{W_1},R_{U_2},R_{W_2})$ is then shown to lie within the four-user MAC region (evaluated with Gaussian inputs) at receiver 0 formed by common and private signals from transmitters 1 and 2. It follows that receiver 0 can decode all the signals $x_0,x_1,x_2$ when using Gaussian inputs, hence the 3-user MAC constraint applies. The calculations are deferred to the appendix.
\end{proof}

We now propose a different scheme that achieves a rate point within a constant of the optimal sum-rate of approximately $3\log \beta$ for any $\beta=2^{2n}$, where $n$ is a positive integer. The restriction of $\beta$ to even powers of two allows to simplify the analysis of the scheme; the scheme itself, as well as the general scheme presented in Section~\ref{s:GaussianManytoOne}, works for arbitrary real-valued channel gains. Consider first only the real-valued channel (assume that $z_i\sim \CalN(0,1)$ and inputs are real-valued). Each user generates a random codebook from a \emph{discrete} distribution 
\begin{equation}\label{e:example_discreteinput}
x_i=\sum_{k=1}^{ \log  \sqrt \beta } x_i(k) 2^{-k},\quad i=0,1,2\,,\end{equation} where the bits $x_i(k)\sim \text{Bernoulli}\left(\frac12\right)$ are i.i.d. over time. In order to show an achievable rate we calculate the single time-step mutual information between input and output for each user,
$$I(x_i;y_i),\quad i=0,1,2\,.$$ Let $\y_i$ denote the noiseless output 
\begin{align*}
  \y_0&=\sqrt \beta x_0+\beta x_1 + \beta x_2
  \\
  \y_1&=\sqrt \beta x_1
  \\ 
  \y_2&=\sqrt \beta x_2\,,
\end{align*}
It is shown in Appendix A.1 of \cite{BT08} that when using inputs such that the outputs $\y_i$ are integer-valued, the additive Gaussian noise $z_i$ causes a loss in mutual information of at most 1.5 bits, i.e. 
\begin{equation}\label{e:noiseloss}I(x_i;\y_i)-1.5\leq I(x_i;\y_i+z_i)=I(x_i;y_i)\,.\end{equation}
Intuitively, this is because $\y_i$ can be recovered from $y_i$ by knowing the value of $[z_i]$, where $[\,\cdot\,]$ is the nearest integer function;  the estimate $H([z_i])\leq 1.5$ allows to show the inequality~\eqref{e:noiseloss}.
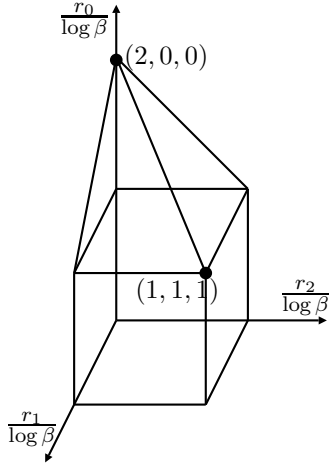
\begin{figure}
\centering
\psset{unit=.7mm,arrowlength=1,arrowinset=0,labelsep=2pt}
\begin{center}
\begin{pspicture}(-14.5,-24)(41,61)
\psline{->}(0,0)(40,0)
\psline{->}(0,0)(0,60)
\psline{->}(0,0)(-13.5,-27)

\psline(25,25)(0,25)
\psline(-8,9)(0,25)


\psline(25,0)(25,25)
\psline(17,9)(25,25)

\psline(25,0)(17,-16)
\psline(17,-16)(-8,-16)
\psline(-8,-16)(-8,9)
\psline(-8,9)(17,9)
\psline(17,9)(17,-16)

\psline(17,9)(0,50)

\psline(-8,9)(0,50)
\psline(25,25)(0,50)

\uput[u](36,0){$\frac{r_2}{\log \beta}$}
\uput[u](-16,-25){$\frac{r_1}{\log \beta}$}
\uput[l](0,57){$\frac{r_0}{\log \beta}$}

\uput[r](0.5,50){\small $(2,0,0)$}
\uput[d](11.6,9){\small $(1,1,1)$}

\pscircle[fillstyle=solid, fillcolor=black](0,49.5){1.2}
\pscircle[fillstyle=solid, fillcolor=black](17,9){1.2}

\end{pspicture}
\end{center}
\caption{The approximate capacity region of the example channel considered in this section. The two dominant corner points are emphasized.}%
\label{fig:ExampleChannelRegion}%
\end{figure}

Note the following key observation: it is possible to perfectly recover the signal $x_0$ from $\y_0$. This follows from writing
$$\y_0=\sqrt\beta x_0+\sqrt\beta[\sqrt\beta(x_1+x_2)]\in \sqrt\beta x_0+\sqrt\beta\mathbb{Z}\,,$$ and the fact that $\sqrt\beta x_0<\sqrt\beta$\,.
Hence $$\sqrt \beta x_0=\y_0\ (\text{mod}\sqrt \beta)\,,$$ and for $i=0,1,2$,
$$I(x_i;y_i)+1.5\geq I(x_i;\y_i)=H(x_i)=\frac12\log\beta\,.$$ 

For the complex-valued channel the same strategy works independently in the complex and real dimension, giving an achievable rate of
$$r_i\geq \log \beta -3,\quad i=0,1,2\,.$$
The sum-rate achieved,
$$r^{\text{lattice}}_{\text{sum}}=3\log\beta -9\approx 3\log \beta \,,$$ is therefore arbitrarily larger than the approximately $2\log\beta$ achieved by the strategy employing Gaussian codebooks. The achievable region for large $\beta$, normalized by $\log \beta$, is depicted in Figure~\ref{fig:ExampleChannelRegion}.

Before proceeding we reflect on why random Gaussian codebooks are suboptimal for the many-to-one channel. Note that the aggregate interference at receiver 0 has support equal to the sumset of the supports of codebooks 1 and 2. As illustrated in Figure~\ref{fig:randomMess}, the sumset of two random (continuously distributed) codebooks fills the space, leaving no room for user 0 to communicate. If each of codebooks 1 and 2 have $m$ points, the sumset can have up to $m^2$ points. In contrast, as illustrated in Figure~\ref{fig:latticeCode}, the sum of two codebooks that are subsets of a lattice looks essentially like one of the original codebooks (and in particular has cardinality $Cm$, where $C$ is a constant independent of $m$). Thus, the cost to user 0 is the same as though due to only one interferer, i.e. the interference is aligned on the signal scale. This theme will reappear throughout the paper.

\begin{figure}
\centering
\includegraphics[width=3in]{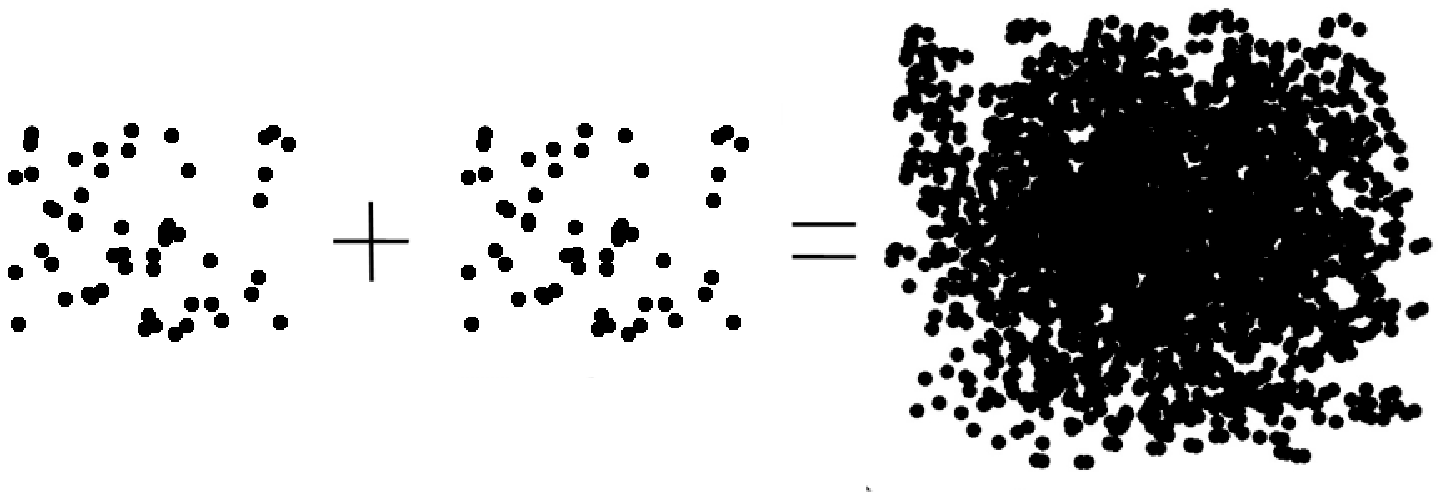}
\caption{The sum of two identical random codebooks with 50 points each. The resulting interference covers the entire space, preventing receiver 0 from decoding.}%
\label{fig:randomMess}%
\includegraphics[width=3in]{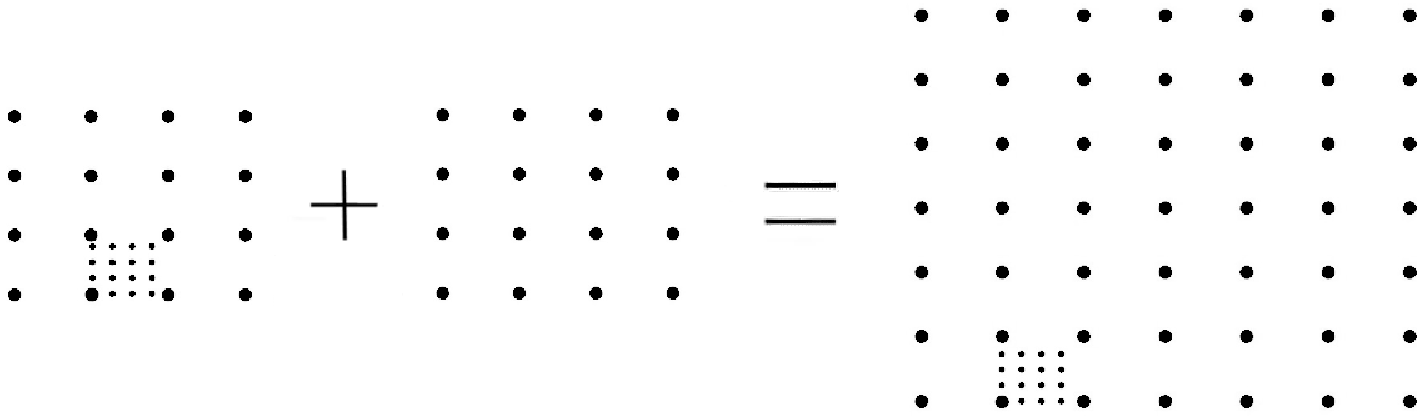}
\caption{User 0 can decode the fine signal in the presence of interference from users 1 and 2. The sum of the interference from users 1 and 2 imposes essentially the same cost as from a single interferer.}%
\label{fig:latticeCode}%
\end{figure}

In order to generalize the intuition gained from this example and provide the framework for finding the capacity of the many-to-one channel to within a constant gap we make use of a deterministic channel model, described in the next section.

%

\section{Deterministic channel model}\label{s:DetModel}

We now present a deterministic channel model analogous to the Gaussian channel. This channel was first introduced in \cite{DeterministicRelay}. 
We begin by describing the deterministic channel model for the point-to-point AWGN channel, and then the two-user multiple-access channel. After understanding these examples, we present the deterministic interference channel.

Consider the model for the point-to-point channel (see Figure~\ref{fig:p2pDeterministic}). The real-valued channel input is written in base 2; the signal---a vector of bits---is interpreted as occupying a succession of levels:
\begin{equation}\label{e:p2p_deterministicBits}
x=0.b_1 b_2 b_3 b_4 b_5\dots\,.
\end{equation}
The most significant bit coincides with the highest level, the least significant bit with the lowest level. The levels attempt to capture the notion of \emph{signal scale}; a level corresponds to a unit of power in the Gaussian channel, measured on the dB scale. Noise is modeled in the deterministic channel by truncation. Bits of smaller order than the noise are lost. The channel may be written as $$y=\lf 2^n x\rf\,,$$ with the correspondence $n=\lf \log \snr \rf$.

Note the similarity of the binary expansion underlying the deterministic model \eqref{e:p2p_deterministicBits} to the discrete inputs \eqref{e:example_discreteinput} in the example channel of the previous section. The discrete inputs \eqref{e:example_discreteinput} are precisely a binary expansion, truncated so that the signal---after scaling by the channel---is integer-valued. Evidently, the achievable scheme for the example channel emulates the deterministic model.

\begin{figure}
\begin{centering}
\psset{unit=1.2mm,arrowlength=1,arrowinset=0,linewidth=.5pt}
\begin{center}
\begin{pspicture}(-10,3)(40,28)
\psframe[linecolor=white](-10,3)(40,28)

\pscircle(0,4){1.5}
\pscircle(0,8){1.5}
\pscircle(0,12){1.5}
\pscircle(0,16){1.5}
\pscircle(0,20){1.5}

\rput(0,4){\footnotesize $b_5$}
\rput(0,8){\footnotesize $b_4$}
\rput(0,12){\footnotesize $b_3$}
\rput(0,16){\footnotesize $b_2$}
\rput(0,20){\footnotesize $b_1$}

\psline{->}(1.5,16)(28.5,12)
\psline{->}(1.5,8)(28.5,4)
\psline{->}(1.5,12)(28.5,8)
\psline{->}(1.5,20)(28.5,16)

\pspolygon(-2,1)(-6,1)(-6,23)(-2,23)

\rput(30,0){
\pscircle(0,4){1.5}
\pscircle(0,8){1.5}
\pscircle(0,12){1.5}
\pscircle(0,16){1.5}
\pscircle(0,20){1.5}
\pspolygon(2,1)(6,1)(6,23)(2,23)
\rput(0,-4){\rput(0,8){\footnotesize $b_4$}
\rput(0,12){\footnotesize $b_3$}
\rput(0,16){\footnotesize $b_2$}
\rput(0,20){\footnotesize $b_1$}}
}

\psline[linestyle=dashed](0,1)(30,1)
\psline(1.5,4)(20.25,1)

\uput[l](-6,12){Tx}
\uput[r](36,12){Rx}
\uput[d](15,0){noise}
\end{pspicture}
\end{center}
\caption{The deterministic model for the point-to-point Gaussian channel. Each bit of the input occupies a signal level. Bits of lower significance are lost due to noise.}
\label{fig:p2pDeterministic}
\end{centering}
\end{figure}
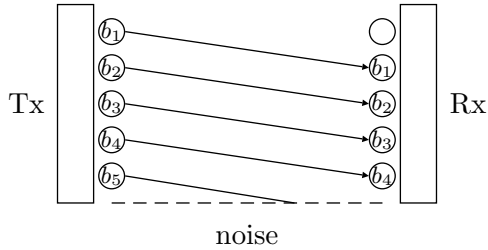

The deterministic multiple-access channel is constructed similarly to the point-to-point channel (Figure~\ref{fig:deterministicMAC}), with $n_1$ and $n_2$ bits received above the noise level from users $1$ and $2$, respectively. To model the superposition of signals at the receiver, the bits received on each level are added {\em modulo two}. Addition modulo two, rather than normal integer addition, is chosen to make the model more tractable. As a result, the levels do not interact with one another.

If the inputs $x_i(t)$ are written in binary, the channel output can be written as
\begin{equation}\label{e:deterministicMACintegerpart}
  y=\lfloor 2^{n_{1}}x_1\rf \oplus \lf 2^{n_{2}}x_2\rfloor 
  \, ,
\end{equation}
where addition is performed on each bit (modulo two) and $\lfloor\, \cdot \, \rfloor$ is the integer-part function.
The channel can also be written in an alternative form, which we will not use in the present paper but leads to a slightly different interpretation. The input and output are $x_1,x_2,y \in \F_2^q$, where $q=\max(n_1,n_2)$. The signal from transmitter $i $ is scaled by a nonnegative integer gain $2^{n_{i}} $ (equivalently, the input column vector is shifted up by $n_{i}$). The channel output is given by \begin{equation}y= \mathbf{S}^{q-n_{1}}x_1+ \mathbf{S}^{q-n_{2}}x_2,\end{equation} where summation and multiplication are in $\F_2$ and $\mathbf{S}$ is a $q\times q$ shift matrix,
\begin{equation}\label{e:shiftMatrix}
\mathbf{S}=\left(
\begin{matrix} 0 & 0 & 0 & \cdots & 0 \cr 1 & 0 & 0 & \cdots & 0 \cr 0 & 1 & 0 & \cdots & 0 \cr \vdots & & & \ddots & \vdots \cr 0 & \cdots & 0 & 1 & 0 \end{matrix}
\right).
\end{equation}

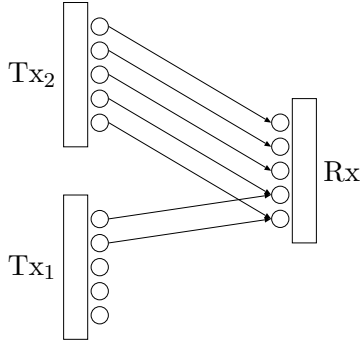
\begin{figure}
\begin{centering}
\psset{unit=.8mm,linewidth=.3pt,arrowlength=1.2,arrowinset=0,labelsep=5pt}
\begin{center}
\begin{pspicture}(-1,12)(37,76)

\rput(0,14){\pscircle(0,0){1.5}
\pscircle(0,4){1.5}
\pscircle(0,8){1.5}
\pscircle(0,12){1.5}
\pscircle(0,16){1.5}
\pspolygon(-2,-4)(-6,-4)(-6,20)(-2,20)

\uput[l](-5,8){$ {\text{Tx}_1}$}
}

\psline{->}(1.5,30)(28.5,34)
\psline{->}(1.5,26)(28.5,30)

\uput[r](35,38){$ \text{Rx}$}


\rput(0,46){
\pscircle(0,0){1.5}
\pscircle(0,4){1.5}
\pscircle(0,8){1.5}
\pscircle(0,12){1.5}
\pscircle(0,16){1.5}
\pspolygon(-2,-4)(-6,-4)(-6,20)(-2,20)
}
\psline{->}(1.5,46)(28.5,30)
\psline{->}(1.5,50)(28.5,34)
\psline{->}(1.5,54)(28.5,38)
\psline{->}(1.5,58)(28.5,42)
\psline{->}(1.5,62)(28.5,46)

\uput[l](-5,54){$ \text{Tx}_2$}

\rput(30,30){
\pscircle(0,0){1.5}
\pscircle(0,4){1.5}
\pscircle(0,8){1.5}
\pscircle(0,12){1.5}
\pscircle(0,16){1.5}
\pspolygon(2,-4)(6,-4)(6,20)(2,20)
}

\end{pspicture}
\end{center} 
\caption{The deterministic model for the Gaussian multiple-access channel. Incoming bits on the same level are added modulo two at the receiver.}
\label{fig:deterministicMAC}
\end{centering}
\end{figure}

An easy calculation shows that the capacity region of the deterministic MAC is 
\begin{equation}
  \begin{split}
    r_1 &\leq n_1 \\
    r_2 &\leq n_2 \\
    r_1 +r_2 &\leq \max(n_1,n_2)\,.
  \end{split}
\end{equation}
Comparing with the capacity region of the Gaussian MAC,\begin{equation}\begin{split} \label{e:MACapprox}
  R_1 &\leq \log (1+\snr_1)\approx \log\snr_1 \\
  R_2 &\leq \log (1+\snr_2) \approx \log \snr_2\\
  R_1+R_2 &\leq \log (1+\snr_1+\snr_2)\approx \max(\log \snr_1,\log \snr_2)\,,
\end{split}\end{equation} we make the correspondence $$n_1=\lf \log \snr_1\rf \quad \text{and} \quad n_2=\lf \log \snr_2\rf\,.$$

\subsection{Deterministic interference channel}\label{sec:deterministicIC}

We proceed with the deterministic interference channel model. Note that the model is completely determined by the model for the MAC.
There are $K +1$ transmitter-receiver pairs (links), and as in the Gaussian case, each transmitter wants to communicate only with its corresponding receiver.  The signal from transmitter $j $, as observed at receiver $i$, is scaled by a nonnegative integer gain $n_{ij} $. 
The channel may be written as
$$y_i=\lf 2^{n_{i0}} x_0\rf \oplus \cdots \oplus \lf 2^{n_{iK}} x_K\rf\,,$$
where, as before, addition is performed on each bit (modulo two) and $\lf\,\cdot\,\rf$ is the integer-part function.

Alternatively, at each time $t $, we may view the input and output, respectively, at link $i$ to be $x_i(t),y_i(t)\in \F_2^q$, where $q=\max_{ij}n_{ij}$.
The channel output at receiver $i$, $0\leq i\leq K$, is given by \begin{equation*}y_i(t)=\sum_{j=0}^K \mathbf{S}^{q-n_{ij}}x_j(t),\end{equation*} where summation and multiplication are in $\F_2$ and $\mathbf{S}$ is a $q\times q$ shift matrix \eqref{e:shiftMatrix}.
The standard definitions of achievable rates and the associated notions are omitted.

The deterministic interference channel is relatively simple, yet
retains two essential features of the Gaussian interference channel:
the loss of information due to noise, and the superposition of
transmitted signals at each receiver. The modeling of noise can be understood through the point-to-point channel above. The superposition of transmitted signals at each receiver is captured by taking the modulo 2 sum of the incoming signals at each level. 

The relevance of the
deterministic model is greatest in the high-$\snr$ regime, where communication is interference---rather than noise---limited; however, we shall see that
even for finite signal-to-noise ratios the deterministic channel
model provides significant insight towards the more complicated
Gaussian model.

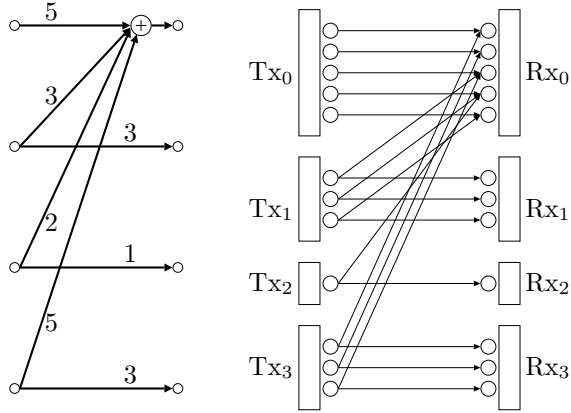
\begin{figure}
\begin{centering}
\psset{unit=.7mm,linewidth=.3pt,arrowlength=1.2,arrowinset=0,labelsep=2pt}
\begin{center}
\begin{pspicture}(-1,2)(90,70)
\rput(56,0){
\rput(0,-8){

\pscircle(0,8){1.5}
\pscircle(0,12){1.5}
\pscircle(0,16){1.5}

\psline{->}(1.5,16)(28.5,16)
\psline{->}(1.5,8)(28.5,8)
\psline{->}(1.5,12)(28.5,12)
\pspolygon(-6,4)(-6,20)(-2,20)(-2,4)}

\uput[l](-6,4){\small${\text{Tx}_3}$}
\uput[r](36,4){\small${\text{Rx}_3}$}

\rput(0,20){
\pscircle(0,0){1.5}
\psline{->}(1.5,0)(28.5,0)
\pspolygon(-6,4)(-6,-4)(-2,-4)(-2,4)
}

\psline{->}(1.5,20)(28.5,56)

\rput(0,24){
\pscircle(0,8){1.5}
\pscircle(0,12){1.5}
\pscircle(0,16){1.5}
\psline{->}(1.5,16)(28.5,16)
\psline{->}(1.5,8)(28.5,8)
\psline{->}(1.5,12)(28.5,12)
\pspolygon(-6,4)(-6,20)(-2,20)(-2,4)}

\uput[l](-6,34){\small$\text{Tx}_1$}
\uput[r](36,34){\small$\text{Rx}_1$}

\psline{->}(1.5,8)(28.5,68)
\psline{->}(1.5,4)(28.5,64)
\psline{->}(1.5,0)(28.5,60)

\psline{->}(1.5,32)(28.5,52)
\psline{->}(1.5,36)(28.5,56)
\psline{->}(1.5,40)(28.5,60)

\rput(0,52){\pscircle(0,0){1.5}
\pscircle(0,4){1.5}
\pscircle(0,8){1.5}
\pscircle(0,12){1.5}
\pscircle(0,16){1.5}
\psline{->}(1.5,4)(28.5,4)
\psline{->}(1.5,16)(28.5,16)
\psline{->}(1.5,8)(28.5,8)
\psline{->}(1.5,12)(28.5,12)
\psline{->}(1.5,0)(28.5,0)
\pspolygon(-6,-4)(-6,20)(-2,20)(-2,-4)
}

\uput[l](-6,60){\small$ \text{Tx}_0$}
\uput[r](36,60){\small$ \text{Rx}_0$}

\rput(30,-8){
\pscircle(0,8){1.5}
\pscircle(0,12){1.5}
\pscircle(0,16){1.5}
\pspolygon(6,4)(6,20)(2,20)(2,4)}

\rput(30,20){
\pscircle(0,0){1.5}
\pspolygon(6,4)(6,-4)(2,-4)(2,4)
}
\uput[l](-6,20){\small$ \text{Tx}_2$}
\uput[r](36,20){\small$ \text{Rx}_2$}

\rput(30,24){
\pscircle(0,8){1.5}
\pscircle(0,12){1.5}
\pscircle(0,16){1.5}
\pspolygon(6,4)(6,20)(2,20)(2,4)}

\rput(30,52){\pscircle(0,0){1.5}
\pscircle(0,4){1.5}
\pscircle(0,8){1.5}
\pscircle(0,12){1.5}
\pscircle(0,16){1.5}
\pspolygon(6,-4)(6,20)(2,20)(2,-4)
}}


\rput(-8,0){
\psline[arrowlength=1,linewidth=.4]{->}(5,0)(34,0)
\psline[arrowlength=1,linewidth=.4]{->}(5,23)(34,23)
\psline[arrowlength=1,linewidth=.4]{->}(5,46)(34,46)
\psline[arrowlength=1,linewidth=.4]{->}(5,69)(26,69)

\psline[arrowlength=1,linewidth=.4]{->}(5,0)(27.2,67.3)
\psline[arrowlength=1,linewidth=.4]{->}(5,23)(26,68.3)
\psline[arrowlength=1,linewidth=.4]{->}(5,46)(26,68.6)

\psline[arrowlength=1,linewidth=.4]{->}(30,69)(34,69)

\pscircle(35,69){1}
\pscircle(4,0){1}
\pscircle(35,46){1}
\pscircle(4,46){1}
\pscircle(4,23){1}
\pscircle(4,69){1}
\pscircle(35,23){1}
\pscircle(35,0){1}
\pscircle(28,69){2}
\psline{-}(27,69)(29,69)
\psline{-}(28,68)(28,70)

\uput{2pt}[u](11,10){\small 5}
\uput{2pt}[u](11,29){\small 2}
\uput{2pt}[u](11,53){\small 3}
\uput{2pt}[u](11,69){\small 5}
\uput{2pt}[u](26,23){\small 1}
\uput{2pt}[u](26,46){\small 3}
\uput{2pt}[u](26,0){\small 3}
}
\end{pspicture}
\end{center} 
\caption{Both figures depict the same channel. On the left is an example of a deterministic many-to-one interference channel with 4 users. The right-hand figure shows how the inputs are shifted and added together (modulo $2$) at each receiver. Each circle on the left side represents an element of the input vector; each circle on the right represents the received signal at a certain level.}
\label{DeterministicZChannelFigure}
\end{centering}
\end{figure}

As in the approach for the Gaussian interference channel, we consider only special cases of the deterministic interference channel: the many-to-one and one-to-many ICs. In the many-to-one IC interference occurs only at receiver 0 (see Figure~\ref{DeterministicZChannelFigure} for an example), and in the one-to-many IC interference is caused by only one user.

\section{Deterministic Many-to-One Interference Channel}\label{s:DeterministicManytoOneCapacity}
In this section we find the capacity region of the deterministic many-to-one IC. By separately considering each level at receiver 0 together with those signals causing interference to the level, the many-to-one channel is seen to be a \emph{parallel} channel, one sub-channel per level at receiver 0. This begs the question: is the capacity of the many-to-one channel equal to the sum of the capacities of the sub-channels? Theorem~\ref{t:DeterministicCapacity} below answer this question in the affirmative.

Some notation is required. First, we can assume without loss of generality that each input $x_i$ is restricted to the elements that appear in the output $y_i$, i.e. $x_i\in \F_2^{n_{ii}}$. Denote by $U_k \subseteq \{1,\ldots,K\}$, $1\leq k\leq n_{00}$, the set of users
potentially causing interference at receiver 0 on level $k$: $U_k=\{i:1\leq i\leq K, n_{0i}-n_{ii}< k\leq n_{0i}\} $.
For a set of users $A\subseteq \{0,1,\ldots,K\} $ and a level $k,1\leq k\leq n_{00}$, denote by
$x_{A|k}$ the vector of signals of users in $A$, restricted to
level $k$ \emph{as observed at receiver 0}. See Figure~\ref{BipartiteRectanglesFigure} for an illustration of these definitions.

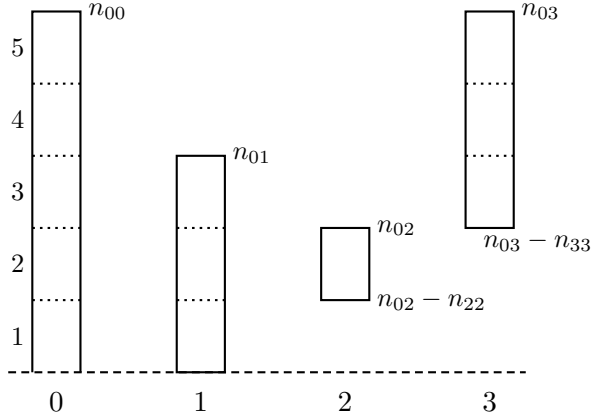
\begin{figure}
\psset{unit=.8mm,arrowlength=1,arrowinset=0,labelsep=3pt}
\begin{center}
\begin{pspicture}(0,2)(87,62)
\psline[linestyle=dashed,dash=3pt 2pt](0,0)(86,0)
\psline(4,0)(4,60)(12,60)(12,0) \rput(8,-5){0}
\pspolygon(28,0)(28,36)(36,36)(36,0)\rput(32,-5){1}
\pspolygon(52,12)(52,24)(60,24)(60,12)\rput(56,-5){2}
\pspolygon(76,24)(76,60)(84,60)(84,24)\rput(80,-5){3}

\psline[linestyle=dashed,dash=1pt 2pt](28,24)(36,24)
\psline[linestyle=dashed,dash=1pt 2pt](28,12)(36,12)
\psline[linestyle=dashed,dash=1pt 2pt](76,36)(84,36)
\psline[linestyle=dashed,dash=1pt 2pt](76,48)(84,48)
\psline[linestyle=dashed,dash=1pt 2pt](4,48)(12,48)
\psline[linestyle=dashed,dash=1pt 2pt](4,36)(12,36)
\psline[linestyle=dashed,dash=1pt 2pt](4,24)(12,24)
\psline[linestyle=dashed,dash=1pt 2pt](4,12)(12,12)

\uput[r](12,60){\small $n_{00}$} \uput[r](36,36){\small
    $n_{01}$}
\uput[r](60,24){\small $n_{02}$} \uput[r](60,12){\small$n_{02}-n_{22}$} \uput[r](84,60){\small
$n_{03}$} \uput[d](88,24.7){\small
$n_{03}-n_{33}$}

\uput[l](4,6){\small $1$} \uput[l](4,18){\small $2$}
\uput[l](4,30){\small $3$} \uput[l](4,42){\small $4$}
\uput[l](4,54){\small $5$}

\end{pspicture}
\end{center}
\caption{The interference pattern as observed at receiver 0 for the channel in Figure~\ref{DeterministicZChannelFigure}. Here, $U_1=\{1\}$, $U_2=\{1,2\}$, $U_3=\{1,3\}$, etc. etc.}\label{BipartiteRectanglesFigure}
\end{figure}

Let $\tilde x_i$ be the restriction of the input from transmitter $i$ to the lowest $(n_{ii} -n_{0i})^+$ levels. This is the part of $x_i $ that does not appear as interference at receiver 0, i.e. this part of the interfering signal is below the noise level. Similarly, let $\hat x_i $ be the restriction of the input from transmitter $i $ to the highest $ (n_{0i} -n_{00 })^+ $ levels. This is the part of $x_i $ that causes interference above the signal level of user 0 (and therefore does not really interfere).
With this notation at our disposal, we are ready to describe the achievable strategy, and then state the capacity region of the deterministic many-to-one IC.

\subsection{Achievable Strategy}

The achievable strategy consists of allocating each level separately, by choosing either user 0 to transmit on a given level, or all users interfering with user 0 to transmit on the level. This scheme aligns the interference as observed by receiver 0, so that several users transmitting on a level inflict the same cost to user 0 as one user transmitting on the same level. Because the scheme considers each level separately, the structure of the achievable region is remarkably simple. The region is next described in more detail. 

First, note that by transmitting on levels that appear above the signal of user 0 or below the noise level as observed by receiver 0, each user can transmit at rate $$r_i\leq f_\text{free}(i),$$ where $f_\text{free}(i)= (n_{0i}-n_{ii})^+ +(n_{0i}-n_{00})^+$, without causing any interference to user 0.  We have that users $\{1,\dots,K\}$ can use rates in the region $$\CalC_{\text{free}}=\{(r_0,\dots, r_K): r_0=0, r_i\leq f_\text{free}(i)\}$$ without causing any interference to user 0. 

For a subset of users $U_k\subseteq \{1,\dots,K\}$, let $\CalC_k$ denote the capacity region of a deterministic many-to-one IC with only one level, and users $U_k$ interfering at receiver 0. Users not in $U_k$, i.e. $\{1,\dots, K\}\setminus U_k$, are not present. It is easy to see that $\CalC_k$ is given by the intersection of the individual rate constraints \begin{align}r_i&\leq 1, \quad i\in U_k\cup \{0\}\label{e:SingleLevelIndividual} \\ r_i&=0,\quad i\notin U_k\nonumber\end{align} and the pair-wise rate constraints
\begin{equation}
\label{e:SingleLevelPairWise}r_0+r_i\leq 1,\quad i\in U_k.
\end{equation} 
The capacity $\CalC_k$ is achieved by time-sharing between two rate points: 1) User 0 transmits a uniformly random bit, while all other users are silent, or 2) User 0 is silent while each user in $U_k$ transmits a uniformly random bit. This is done for each level $1\leq k\leq \an$.

Let $\underC$ be the set of rate points achieved by our scheme. Since the achievable scheme treats each level separately, the achievable region is the sum of the regions for each level and the set of points achievable without causing any interference: \begin{equation}
\underC = \CalC_{\text{free}}+\sum_{k=1}^\an \CalC_k.
\label{e:AchievableRegionParallel}\end{equation}

\subsection{Outer Bound}
We now turn to the outer bound, first rewriting the constraints on each level. Consider some level $k$, $1\leq k\leq \an$. For any set of users $\S\subseteq\{1,\dots,K\}$, we can form a sum-rate constraint on the users $\S\cup\{0\}$ by adding a single pairwise constraint on $r_0+r_i$ for some $i\in \S$ together with individual rate constraints on users $\S\setminus\{i\}$ \eqref{e:SingleLevelIndividual}:
\begin{equation}r_0+\sum_{i\in (U_k\cap \S)}r_i\leq f_k(\S),\label{e:SingleLevelSumRate}\end{equation} where \begin{equation}f_k(\S)=\max(|U_k\cap \S|,1).\label{e:f_k}\end{equation} 

The following lemma gives an outer bound $\overC$ to the capacity region. Thus the capacity region, $\CalC_D$, of the $K+1$ user deterministic many-to-one IC, is bounded as  $$\underC\subseteq \CalC_D\subseteq \overC.$$

\begin{lemma}\label{l:outerBound}
$\CalC_D$ is contained in $\overC$, where $\overC$ is given by the intersection of the individual rate constraints \begin{equation}\label{individualConstraint}
   r_i\leq {\al_ {ii}}, \quad 0\leq i\leq K,
\end{equation} and the $2^K-1$ sum-rate constraints \begin{equation}\label{sumRateConstraint}
   r_0+\sum_{i\in \S}r_i\leq f_\text{free}(\S)+\sum_{k=1}^{\al_{00}}f_k(\S),\quad \S\subseteq \{1,\ldots,K\},\S\neq \varnothing\end{equation}
where
$f_k(\S)$ is defined above in equation \eqref{e:f_k} and $f_\text{free}(\S)=\sum_{i\in \S} f_\text{free}(i)$.
\end{lemma}
 
The bound in the lemma is tight, as shown in the following theorem. Thus, the capacity region is equal to the sum of the capacities of the sub-channels.

\begin{theorem} \label{t:DeterministicCapacity}
The achievable region is equal to the outer bound, i.e. $$\underC=\overC=\CalC_D.$$\end{theorem}


We first prove the constraints in equations (\ref{individualConstraint}) and (\ref{sumRateConstraint}) characterizing $\overC$, and then show that the region coincides with the achievable region $\underC$ of equation \eqref{e:AchievableRegionParallel}.

\begin{proof}[Proof of Lemma~\ref{l:outerBound}]
Clearly, the rate across each link cannot exceed the point-to-point capacity; hence
\begin{equation}r_i\leq {\al_ {ii}} \quad 0\leq i\leq K.\end{equation}

Next, we prove a sum-rate constraint on an arbitrary set of users $\S\cup \{0\}$, where $\S\subseteq \{1,\ldots,K\}$. We give the following side information  to receiver 0: at each level
$k,1\leq k\leq \al_{00}$, the input signals of all interfering users
in $\S$ except for one, and also $\{\hat x_i\}_{i\in\S} $ and the inputs of all users not
in $\S $. More precisely, for each $k$, $1\leq k\leq \al_{00}$, let
$Q_k$ be any set that satisfies $Q_k\subseteq (U_k\cap \S)$ and
$|Q_k|=(|U_k\cap \S|-1)^+$ . We give the side information \begin{equation}\label{e:manyToOneSideInf}s_0 =
\left(\{x_{Q_k|k}\}_{k=1}^\an, \{x_i\}_{i\notin
\S},\{\hat x_i\}_{i\in\S} \right).\end{equation}

Recall that $y_{i,k}=x_{i,k}$ for users $i\neq 0$, hence $I(y_i^N;x_i^N)=H(x_i^N)$. Fano's inequality, the data processing inequality, the chain rule for mutual information, independence of $x_0$ and $s_0$, and breaking apart the signals according to level gives
\begin{align} N(r_0 +\sum_{i\in \S} r_i-\eps_N)\nonumber
&\leq     I(y_0^N,s_0^N;x_0^N)+\sum_{i\in \S}
   I(y_i^N;x_i^N)\nonumber
\\&=     I(y_0^N;x_0^N|s_0^N)+    \sum_{i\in \S}
   I(y_i^N;x_i^N)\nonumber
\\ &=   H(y_0^N|s_0^N)-H(y_0^N|s_0^N,x_0^N)+ \sum_{i\in \S}
   H(y_i^N)\nonumber
\\ &=   H\bigg(\bigg\{x^N_{0|k}+\sum_{i\in
   U_k}x^N_{i|k}\bigg\}_{k=1}^\an, \sum_{i=1}^K \hat x^N_i \bigg|\{x^N_{Q_k|k}\}_{k=1}^\an, \{x^N_i\}_{i\notin
   \S},\{\hat x^N_i\}_{i\in \S}\bigg)\nonumber
   \\ &\quad - H\bigg(\bigg\{\sum_{i\in
       U_k}x^N_{i|k}\bigg\}_{k=1}^\an, \sum_{i=1}^K\hat x^N_i \bigg|\{x^N_{Q_k|k}\}_{k=1}^\an, \{x^N_i\}_{i\notin
       \S},\{\hat x^N_i\}_{i\in \S}\bigg)\nonumber
   \\&\quad        +
      H\left(\{x^N_{Q_k|k}\}_{k=1}^\an,\{x^N_{U_k\cap \S\setminus Q_k|k}\}_{k=1}^\an,\{\hat x^N_i\}_{i\in \S},\{\tilde x^N_i\}_{i\in \S}\right)\nonumber\,.
      \end{align}
Continuing, the fact that $s_0$ is independent of $x_0$, removing conditioning, the chain rule for mutual information, and the
independence bound on entropy justify the remaining inequalities:      
      \begin{align}
&\leq  H\left(\bigg\{x_{0|k}^N+
   x_{U_k\cap \S\setminus Q_k|k}^N\bigg\}_{k=1}^\an\right)\nonumber   -H\left(\{ x^N_{U_k\cap \S\setminus
   Q_k|k}\}_{k=1}^\an\big|\{\hat x_i\}_{i\in\S}\right)
   \\&\quad + H\left(\{x^N_{Q_k|k}\}_{k=1}^\an\right)+   H\left(\{x^N_{U_k\cap \S\setminus
       Q_k|k}\}_{k=1}^\an\big|\{\hat x_i\}_{i\in\S}\right)+H\left(\{\hat x^N_i\}_{i\in \S}\right)+H\left(\{\tilde x^N_i\}_{i\in \S}\right)\nonumber
\\ &\leq    Nn_{00}+N\sum_{k=1}^\an
   (|U_k\cap \S|-1)^+ \nonumber
 +N\sum_{i\in \S}\left( (n_{0i}-n_{00})^++(n_{ii}-n_{0i})^+
  \right)\nonumber
\\&=N\left(f_\text{free}(\S)+\sum_{k=1}^{\al_{00}}f_k(\S)\right)\nonumber.
\end{align}
Taking $N\to \infty$ proves the sum-rate constraint.
\end{proof}

\subsection{Proof of Theorem~\ref{t:DeterministicCapacity}}
The preceding lemmas give algebraic characterizations of $\underC$ and $\overC$; therefore, the result of Theorem~\ref{t:DeterministicCapacity} is essentially an algebraic property. We begin the proof by taking another look at the achievable region. 

As outlined above, the achievable strategy consists of allocating each level $k,1\leq k\leq \an$, entirely to user 0 or to all users in $U_k$ (recall that $U_k$ is the set of users potentially causing interference to user 0 on level $k$).
Now, the sum-rate constraint \eqref{sumRateConstraint} on a single set of users $\S$ can be met with equality by 1) having each user in $\S$ transmit on levels not causing interference to user 0, and 2) if $|U_k\cap \S|\geq 2$, then users in $\S$ transmit on level $k$ while user 0 is silent, if $U_k\cap \S=\varnothing$, then user 0 transmits on level $k$ while all other users are silent, and if $|U_k\cap \S|=1$ then either user 0 or the interfering user transmits on level $k$.


These rules for allocating levels can be understood through a bipartite graph (see Figure~\ref{BipartiteGraphFigure}). The left-hand set of vertices is indexed by the subsets $\S\cup \{0\}$ for each nonempty $\S\subseteq \{1,\ldots,K\}$ and also a vertex for each user $0,1,\ldots, K$; the right-hand set of vertices is indexed by the levels $k,1\leq k\leq \an$. There is a solid edge between $\S\cup\{0\}$ and $k$ if $|U_k\cap \S|\geq 2$, signifying that for our scheme to achieve the constraint on  $\S\cup\{0\}$ with equality, it is required that all users in $\S$ other than user 0 transmit on level $k$. There is a dashed edge between  $\S\cup\{0\}$ and $k$ if $U_k\cap \S=\varnothing$, signifying that no users other than user 0 may transmit on level $k$. There is no edge if $|U_k \cap \S|=1$.
Finally, for each of the $K+1$ individual constraints (on user $i$, $0\leq i\leq K$), the vertex labeled $i$ has solid edges to those $k$ with $i\in U_k$ (and no other edges), signifying that user $i$ must fully use all available levels.

For a vertex $v$ on the left-hand side, let $N_\b(v)$ be the set of (right-hand side) vertices connected by solid edges to $v$, and similarly, let $N_\r(v)$ be the set of (right-hand side) vertices connected by dashed edges to $v$.

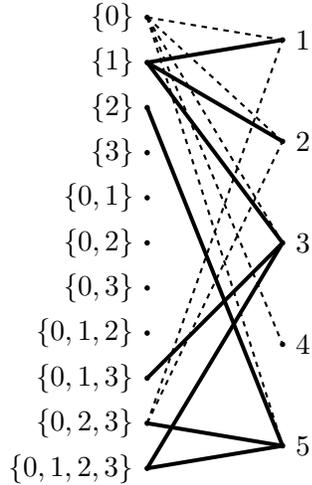
\begin{figure}
\psset{unit=.6mm,arrowlength=1,arrowinset=0}
\begin{center}
\begin{pspicture}(-1,-2)(31,104)

\psline[linestyle=dashed,dash=2pt 2pt](0,100)(30,5)
\psline[linestyle=dashed,dash=2pt 2pt](0,100)(30,27.5)
\psline[linestyle=dashed,dash=2pt 2pt](0,100)(30,50)
\psline[linestyle=dashed,dash=2pt 2pt](0,100)(30,72.5)
\psline[linestyle=dashed,dash=2pt 2pt](0,100)(30,95)

\psline[linewidth=1.5pt](0,90)(30,72.5)
\psline[linewidth=1.5pt](0,90)(30,50)
\psline[linewidth=1.5pt](0,90)(30,95)

\psline[linewidth=1.5pt](0,0)(30,50)
\psline[linewidth=1.5pt](0,0)(30,5)

\psline[linestyle=dashed,dash=2pt 2pt](0,10)(30,95)
\psline[linestyle=dashed,dash=2pt 2pt](0,10)(30,72.5)
\psline[linewidth=1.5pt](0,10)(30,5)

\psline[linewidth=1.5pt](0,20)(30,50)

\psline[linewidth=1.5pt](0,80)(30,5)

\pscircle*(0,100){.6} \pscircle*(0,90){.6} \pscircle*(0,80){.6}
\pscircle*(0,70){.6} \pscircle*(0,60){.6} \pscircle*(0,50){.6}
\pscircle*(0,40){.6} \pscircle*(0,30){.6} \pscircle*(0,20){.6}
\pscircle*(0,10){.6} \pscircle*(0,0){.6}

\uput[l](0,100){$\{0\}$}\uput[l](0,90){$\{1\}$}\uput[l](0,80){$\{2\}$}
\uput[l](0,70){$\{3\}$} \uput[l](0,60){$\{0,1\}$}
\uput[l](0,50){$\{0,2\}$} \uput[l](0,40){$\{0,3\}$}
\uput[l](0,30){$\{0,1,2\}$} \uput[l](0,20){$\{0,1,3\}$}
\uput[l](0,10){$\{0,2,3\}$} \uput[l](0,0){$\{0,1,2,3\}$}

\pscircle*(30,5){.6} \pscircle*(30,27.5){.6}
\pscircle*(30,50){.6}\pscircle*(30,72.5){.6}\pscircle*(30,95){.6}

\uput[r](30,5){5} \uput[r](30,27.5){4} \uput[r](30,50){3}
\uput[r](30,72.5){2} \uput[r](30,95){1}

\end{pspicture}
\end{center}
\caption{Bipartite graph associated with the interference pattern in Figure~\ref{BipartiteRectanglesFigure}. For clarity, only the edges adjacent to the top three and bottom three vertices on the left-hand side were included.}\label{BipartiteGraphFigure}
\end{figure}

\begin{definition}Given a set of constraints, a subset $A$ of these constraints is said to be
\emph{consistent} if there is a point that lies on each constraint in $A$ simultaneously, and the point does not violate any of the other constraints.\end{definition}

\begin{definition} Given a set of constraints with bipartite graph as described above, a collection of constraints is said to be \emph{compatible} if for any two of the constraints on sets $\S$, $\S^\pr$, it holds that $N_\b(\S) \cap N_\r(\S^\pr)=\varnothing$ and $N_\b(\S^\pr) \cap N_\r(\S)=\varnothing$.\end{definition}

The next lemma is immediate from the definitions.
\begin{lemma} It is possible to achieve at least one point in the intersection of the hyperplanes defining any collection of compatible constraints.\label{l:achieveCompatible}
\end{lemma}
\begin{proof}
It is necessary to check that the assignment for achieving each constraint individually works for the collection of compatible constraints simultaneously. To see this, note that if a set $A$ of constraints (indexed by the sets of users) are compatible, it must be that $$\left(\bigcup_{\S\in A}  N_\b(\S) \right)\bigcap\left( \bigcup_{\S\in A} N_\r(\S)\right) =\varnothing.$$ Thus, in the graph induced by constraints in $A$, each vertex on the right-hand side of the graph has only dashed edges or only solid edges (or no edges), i.e. the assignments agree and all constraints can be achieved simultaneously. This proves the lemma.
\end{proof}

To finish the proof of Theorem~\ref{t:DeterministicCapacity}, we show in the next lemma that if a collection of constraints is consistent, then it is also compatible. In other words, the corner points of the outer bound polyhedron are compatible, and hence by Lemma~\ref{l:achieveCompatible} achievable. 

\begin{lemma} If a collection of constraints  from Theorem~\ref{t:DeterministicCapacity} is consistent, then it is also compatible.
 \label{deterministicAchievable}
\end{lemma}
\begin{proof}
The proof is deferred to the Appendix.
\end{proof}


The proof of Theorem~\ref{t:DeterministicCapacity}, which gives the capacity region of the many-to-one deterministic IC, now requires only a straightforward application of the previous lemmas.

Consider any corner point of the outer bound polyhedron. It is located at the intersection of $K+1$ consistent constraints, and this point is achievable by the previous two lemmas.
Hence all corner points of the outer bound polyhedron are achievable, and because it is convex, the polyhedron defined by all the constraints is the capacity region of the channel.

\begin{remark}  It is a pleasing feature of this channel that all corner points of the capacity region can be achieved with zero probability of error, and using a fixed code book with inputs i.i.d over time.
\end{remark}

\begin{remark} There is a natural generalization of the Han and Kobayashi scheme from the two-user interference channel to the many-user interference channel. The capacity-achieving strategy for the deterministic many-to-one IC presented in this section falls within this generalized class, with each user's signal consisting entirely of private information. 
\end{remark}

\section{Approximate Capacity Region of the Gaussian Many-to-One Interference Channel}\label{s:GaussianManytoOne}

In this section we present inner and outer bounds to the capacity region of the Gaussian many-to-one IC, analogous to those proved for the deterministic case. However, unlike in the deterministic case, the inner and outer bounds do not match: there is a gap of approximately $5 K \log K$ bits per user (there are $K+1$ users). In comparing the inner and outer bounds, we make use of the deterministic capacity result from the previous section. The achievable region and outer bound, in turn, are shown to lie within $ 3 K \log K$ and $2K\log K$, respectively, bits per user of the capacity region of an appropriately chosen deterministic channel.

In order to harness the understanding gained from the deterministic channel toward the Gaussian case, we construct a similar diagram as that used earlier to describe the signal observed at receiver 0 (Figure~\ref{GaussianRectangles}). Recall the notation $\snr_i=|h_{ii}|^2 P_i/N_0$ for $0\leq i\leq K$, and $\inr_{ij}=|h_{ji}|^2 P_i/N_0$ for $0\leq i,j\leq K$. Since interference occurs only at receiver 0, we shall write  $\inr_i$ instead of $\inr_{i0}$.
For convenience, assume w.l.o.g. that the users are ordered so that $\inr_i/\snr_i\leq \inr_{i+1}/\snr_{i+1}$ for $1\leq i\leq K-1$.
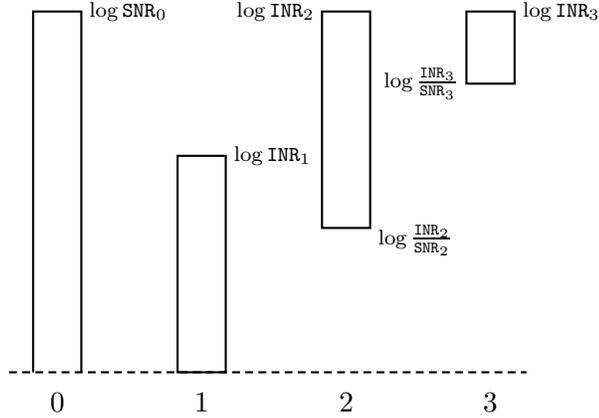
\begin{figure}
\begin{centering}
\psset{unit=.8mm,arrowlength=1,arrowinset=0,labelsep=3pt}
\begin{center}
\begin{pspicture}(-1,-2)(87,62)
\psline[linestyle=dashed,dash=3pt 2pt](0,0)(86,0)
\psline(4,0)(4,60)(12,60)(12,0) \rput(8,-5){0}
\pspolygon(28,0)(28,36)(36,36)(36,0)\rput(32,-5){1}
\pspolygon(76,48)(76,60)(84,60)(84,48)\rput(56,-5){2}
\pspolygon(52,24)(52,60)(60,60)(60,24)\rput(80,-5){3}

\uput[r](12,60){\scriptsize $\log \snr_0$} 
\uput[r](36,36){\scriptsize $\log \inr_1$}
\uput[r](84,60){\scriptsize $\log \inr_3$} 
\uput[l](76,48){\scriptsize{$\log\frac{\inr_3}{\snr_3}$}} 
\uput[l](52,60){\scriptsize $\log\inr_2$} 
\uput[r](60,22){\scriptsize $\log\frac{\inr_2}{\snr_2}$}

\end{pspicture}
\end{center}
\caption{This figure is analogous to Figure~\ref{BipartiteRectanglesFigure}, and shows the interference pattern as observed by receiver 0 (for a different choice of channel gains).}\label{GaussianRectangles}
\end{centering}
\end{figure}

\subsection{Achievable Region}
The achievable strategy mimics the strategy for the deterministic channel, generalizing the scheme proposed for the example channel in Section~\ref{s:ExampleManytoOne}. It can be summarized in a few key steps. First, the range of power-to-noise ratios at receiver 0 is partitioned into intervals to form levels, like in the deterministic channel. There is an independent lattice code for each level, chosen in such a way that the levels do not interact. The scheme then reduces to the achievable scheme for the deterministic channel (with different rates on each level). 

\begin{remark}In using a random lattice instead of the binary expansion, the construction is seemingly different from the one used for the example channel; yet the binary expansion is also a lattice, and both schemes partition the power-to-noise ratios into levels. A direct generalization of the example scheme using binary inputs is also possible; such an approach is not pursued here because it leads to a larger gap from the outer bound and also requires a more technical development (see \cite{BT08}, where a direct approach is taken for the two-user interference channel).
\end{remark}

We now describe the achievable scheme in detail. The power range as observed at receiver 0 is partitioned according to the values $\inr_i$ and $\frac{\inr_i}{\snr_i}$ for all users $i$, $1\leq i\leq K$. More precisely, let $\snr_0=v_1$ and for $1\leq i\leq K$ let $v_{2i}=\inr_i$ and $v_{2i+1}=\frac{\inr_i}{\snr_i}$. Next, remove elements of $\{v_1,\dots,v_{2K+1}\}$ of magnitude less than 1, i.e. let $\{u_1,\dots,u_M\}=\{v_i: v_i>1\}$. Denote by $u_{(k)}$ the $k$th smallest value among $\{u_1,\dots,u_{M}\}$, and let $q_k=u_{(k)}$ for $k\geq 1$, and $q_{-1}=0$, $q_0=1$. The highest endpoint is $q_M=\max(\snr_0,\max_k \inr_k)$. The resulting intervals are $[q_{k-1},q_{k}]$, $0\leq k\leq M$. The partition of power ranges into intervals plays the role of levels in the deterministic channel.

A signal power $\theta_k$, to be specified later, is associated with each level. 
Each user $i$, $0\leq i\leq K$, decomposes the transmitted signal into a sum of independent components 
$$x_i=\sum_{k=0}^M X_i(k),$$ component $X_i(k)$ being user $i$'s input to the $k$th level. The signal $X_i(k)$ has power $\theta_k / |h_{0i}|^2$, so is observed by receiver 0 to be of power $\theta_k$. Of course, each user must satisfy an average power constraint, so does not transmit on higher levels than the power constraint allows: $X_i(k)\equiv 0$ for $k>k_{\text{max}}(i)$, where $q_{k_{\text{max}}(i)}=\inr_i$ for $1\leq i\leq K$ and $q_{k_{\text{max}}(0)}=\snr_0$. Also, user 0 does not transmit on level 0, losing at most 1 bit. 

For each interval $[q_{k-1},q_{k}]$, a lattice code is selected, as in \cite{Loeliger:AveragingBounds}: the spherical shaping region has average power per dimension $\theta_k$ and the lattice is good for channel coding. The rate $R_k$ of the lattice is chosen to allow decoding.
All users transmitting on a given level use the same code (with independent dithers). 
As in the deterministic channel, for each level, either user 0 transmits or all of the interfering users transmit. 

We next describe the decoding procedure at receiver 0. Decoding occurs from the top level downwards, treating the signals from lower levels as Gaussian noise. When the signal on a level is decoded, it is subtracted off completely, and decoding proceeds with the next highest level. Therefore, in describing the decoding procedure, we inductively assume all higher levels have been correctly decoded. On levels where user 0 is silent and interfering users transmit, only the aggregate interfering signal on the level is decoded. This is accomplished by decoding to the nearest lattice point. 

The probability of error analysis is simple, because the sum of subsets of an infinite lattice constellation results in a subset of the same infinite constellation. Furthermore, the probability of decoding error when using lattice decoding does not depend on the transmitted codeword. Thus, because each user transmitting on a level uses a subset of the same infinite lattice, it suffices to consider the decoding of an arbitrary codeword from the lattice. Theorem 7 of \cite{Loeliger:AveragingBounds} shows that if the rate (density of lattice points) is not too high, then receiver 0 is able to decode the sum. The following is a special case discussed immediately following the more general result of Theorem 7:

\begin{theorem}[\cite{Loeliger:AveragingBounds}] \label{t:LatticeDecoding}
  Arbitrarily reliable transmission at rate $R$ is possible with lattice codes of the form $(v+\Lambda)\cap S$, provided $$R< \log \left(\frac{P}{N}\right)\, .$$ Here $\Lambda\subset \R^N$ is a lattice, $v\in \R^N$ is a dither (i.e. shift), $S$ is a spherical shaping region with power $P$ per dimension, and $N$ is the noise variance per dimension.
\end{theorem}

It remains to specify the powers, $\theta_k$, and the rates, $R_k$, for each level. Denote by $N_0(k)$ the variance of all signals on levels $0,\dots,k-1$ plus the additive Gaussian noise as observed at receiver 0: $$N_0(k)=N_0+\sum_{i=0}^{K}\sum_{j=0}^{k-1}\E|X_i(j)|^2|h_{i0}|^2 \leq N_0+\sum_{i=0}^{K}\sum_{j=0}^{k-1}\theta_j\,.$$
The rates achieved by each user transmitting on level $k$ will be \begin{equation}\label{e:RatePerLevel}R_k=\log \left(\frac{\theta_k}{N_0(k)}\right)^+,\end{equation} so that the probability of decoding error vanishes by Theorem~\ref{t:LatticeDecoding}.
We must choose $\theta_k$ in order that the average power constraint is satisfied. Let \begin{equation}\label{e:manyToOne_User0power}\theta_k=(q_k-q_{k-1})N_0\,.\end{equation}
User 0 uses power $\theta_k/|h_{00}|^2$ to transmit on level $k\geq 0$ with $q_k\leq \snr_0$, so that the received power is $\theta_k$. By the definitions, the total power used by user 0 is at most $P_0$. Similarly, to transmit on level $k$ user $i$ uses power 
\begin{equation}\label{e:mayToOne_transmitPower}\theta_k /|h_{0i}|^2.\end{equation} 

Thus, the signals $X_i(k)$, $1\leq i\leq K$ are observed at receiver 0 at power $\theta_k$. Now, we can upper bound $N_0(k)$ by assuming all users other than user $0$ transmit on all levels of lower index. This gives the crude bound $N_0(1)\leq (K+1) N_0$ and for $k> 1$ 
\begin{equation}\label{e:manyToOne_noiseBoundUser0}N_0(k)\leq K q_{k-1}N_0.\end{equation} It must be noted that users other than $0$ will in general have the Gaussian noise at some other power than $N_0/|h_{0i}|^2$; however, since these users only transmit at levels above the noise level, the only source of noise when decoding the lowest level is the additive Gaussian noise. Hence for a user $i\neq 0$ it holds that \begin{equation}\label{e:manyToOne_noiseBound}N_i(k)\leq N_0 q_{k-1}/|h_{0i}|^2, \quad k\geq 1.\end{equation} Recall that we have assumed that for all users $1\leq i\leq K$, $\inr_i>1$. From this, the choice of powers $\theta_k$ \eqref{e:manyToOne_User0power} and \eqref{e:mayToOne_transmitPower}, and the estimates \eqref{e:manyToOne_noiseBoundUser0} and \eqref{e:manyToOne_noiseBound},
we have that the rate of the codebook for level $k> 1$ can be taken as \begin{equation}\begin{split}R_k=\log\left(\frac{\theta_k}{N_0(k)}\right)^+ &\geq \log\left(\frac{q_k-q_{k-1}}{q_{k-1}}\right)^+-\log K\\ &\geq \log\left(1+\frac{q_k-q_{k-1}}{q_{k-1}}\right)-1-\log K
\\&=\log q_k-\log q_{k-1}-1-\log K \,.\label{e:LatticeRate}\end{split}\end{equation}

To compare with the achievable region for the deterministic channel, we make the correspondence $$n_{ii}=\log\snr_i,\quad  0\leq i\leq K\,,\quad \text{and}\quad n_{0i}=\log\inr_i, \quad 1\leq i\leq K\,.$$ Let $l_k$ be the ordered version of the set of endpoints of intervals $\left\{\{n_{0i},n_{0i}-n_{ii}\}_{i=1}^K, 0, n_{00}\right\}$, i.e. $l_k=\log q_k$. Recall the notation $k_{\text{max}}(0)$ is the highest level that user $0$ can use, so that $q_{k_{\text{max}}(0)}=\snr_0$, and also $M$ is the total number of levels. 

We can now finish describing the achievable strategy for the Gaussian channel. On levels without user 0 present, i.e. $k=0$ or $k>k_{\text{max}}(0)$, all users use the full available rate, i.e. for $k=0$ user $i$ gets rate at least $$\log \left(\frac{\snr_i}{\inr_i}\right)^+=(n_{ii}-n_{0i})^+$$ and on the levels $k>k_{\text{max}}(0)$ user $i$ gets rate 
\begin{align*}
  \log \left(\frac{\inr_i}{\snr_0}\right)^+-(M-k_{\text{max}}(0))\log K=(n_{0i}-n_{ii})^+-(M-k_{\text{max}}(0))\log K\,.
\end{align*} In other words, the region $$\CalC_{\text{free}}-(M-k_{\text{max}}(0))\log K (1,1,\dots,1)$$ is achievable without any further constraints on the rates of users on levels $1\leq k\leq k_{\text{max}}(0)$.

Now, each level $k$ with $1\leq k\leq k_{\text{max}}(0)$ (user 0 is present on these levels) can support the rate points $(R_k,0,\dots,0)$ and $\{r_i=R_k:i\in U_k\}\cup\{r_i=0:i\notin U_k\}$, i.e. restricting attention to level $k$, the region $$R_k C_{l_k}$$ is achievable, where $C_j$ is the capacity of a deterministic many-to-one IC with a single level, restricted to users $\{0\}\cup U_j$, given in \eqref{e:SingleLevelIndividual} and \eqref{e:SingleLevelPairWise}. Note that by the definition of $\{l_k\}$, the regions $C_j$ are the same for $l_{k-1}<j\leq l_k$.
Thus, rewriting the rate $R_k$ \eqref{e:LatticeRate} as $$R_k\geq l_k-l_{k-1}-1-\log K\,,$$  the achievable region restricted to levels  $1\leq k\leq k_{\text{max}}(0)$ is \begin{align*}&\sum_{k=1}^{k_{\text{max}}(0)} R_k C_{l_k}\supseteq\sum_{k=1}^{k_{\text{max}}(0)} \sum_{j=l_{k-1}+1}^{l_{k}}C_j-k_{\text{max}}(0)\log K (1,1,\dots,1)\\&=\sum_{j=1}^{n_{00}}C_j-k_{\text{max}}(0)\log K (1,1,\dots,1)\,.\end{align*}
Adding to the region from the previous paragraph, we see that the achievable region contains the region 
\begin{equation}\label{e:AcheivableGaussDetComparison}
\CalC_{\text{free}}+\sum_{j=1}^{n_{00}}C_j-M\log K(1,1,\dots,1)\,,\end{equation}
which is exactly the deterministic capacity region \eqref{e:AchievableRegionParallel}, up to a gap of at most $(M+1)\log K$ bits per user. But $M\leq 2K+1$ since there are $2K+2$ total endpoints including those of user 0's signal, so the gap is no greater than $(2K+1)\log K$ bits per user. 
\begin{remark}
  The fact that the gains $n_{ij}$ are restricted to be integer-valued in the deterministic channel has been disregarded in the above argument. However, this does not pose a problem: instead of putting $n_{ij}=|h_{ij}|^2 P_j/N_0$, one may scale by a sufficiently large integer $T$ and set $n_{ij}=\lf T|h_{ij}|^2 P_j/N_0\rf$, and normalize by $T$. The result is that \eqref{e:AcheivableGaussDetComparison} is simply replaced by the same expression minus $\eps$, where $\eps$ is an aribtrary constant greater than zero. An important point is that the achievable region itself has been set; in this section the capacity of the deterministic channel is only used to relate two algebraic quantities.
\end{remark}

We now turn to the outer bound. 

\subsection{Outer Bound}

We attempt to emulate the proof of the outer bound for the deterministic case, where we gave receiver 0 side information consisting of all but one of the interfering signals at each level. Continuing with the analogy that additive Gaussian noise corresponds to truncation in the deterministic channel, we introduce independent Gaussian noise with appropriate variance in order to properly restrict the side information given to receiver 0. For example, if $\inr_i=p$ and $\inr_{i-1}/\snr_{i-1}=q$, then giving the part of the signal $x_i$ above $q$ as side information to receiver 0 calls for $s=x_i+w_i$ where $w_i\sim \CalCN(0,q N_0)$. Use of this idea leads to the outer bound of the following lemma. 

\begin{lemma}\label{sumRateLemma}The capacity region of the Gaussian many-to-one IC is bounded by each of the individual constraints
$$r_i\leq \log (1+\snr_i),\quad 0\leq i\leq K.$$
Moreover, for each $\S\subseteq \{1,\ldots,K\}$ with the property that a relabeling of the indices of $\S$ allows $ \S =\{1,\ldots,m\} $ (where $m=|\S|$) such that
\begin{equation}\begin{split}
\label{conditions}
&\snr_0>1,\quad \frac{\inr_m}{\snr_m}\leq \snr_0, 
\quad
\inr_i>1,\quad 1\leq i\leq m
\\& \frac{\inr_i}{\snr_i}\leq \frac{\inr_{i+1}}{\snr_{i+1}},\quad \inr_i<\inr_{i+1},\quad  1\leq i\leq m-1\,,
\end{split}
\end{equation}
the following sum-rate constraint holds:
\begin{equation}\begin{split}\label{simpleSumRate}
r_0+ r_1 +\dots +r_m  &\leq
   \sum_{i=1}^m
   \log\bigg(\frac{\snr_i}{\inr_i}\bigg)^+  
   +\sum_{i=1}^{m-1} \left(\log(\inr_i)- \log\left(
   \frac{\inr_{i+1}}{\snr_{i+1}}\right)^+\right)^+ 
\\ &\quad
   +\max(\log(\inr_m),\log(\snr_0))
   +( m + 2 ) \log( m + 1 ).
\end{split}\end{equation}
\end{lemma}
\begin{proof}
The proof is deferred to the appendix.
\end{proof}
\begin{remark}The conditions \eqref{conditions} do not nullify any useful constraints. If $\snr_0\leq 1$, then $r_0\leq 1$ (from the point-to-point constraint), and the capacity region is essentially (within one bit per user) given by the intersection of the individual rate constraints. The other conditions ensure that a user causes meaningful interference to receiver 0, and should therefore be included in the constraint: if $\frac{\inr_m}{\snr_m}> \snr_0$ then the signal from user $m$ may be subtracted off by receiver 0 before attempting to decode the intended signal (user $m$ must reduce the rate by at most $\log K$ bits for this to be true); if the signal from transmitter $i$ has $\inr_i\leq 1$, then transmitter $i$ may just transmit at the full available power, causing essentially (again up to $1$ bit) no interference to user 0. The choice 
$\frac{\inr_i}{\snr_i}\leq \frac{\inr_{i+1}}{\snr_{i+1}}$ is simply a relabeling of the users; with this labeling, if $\inr_i\geq\inr_{i+1}$, then user $i+1$ may be removed from the sum-rate constraint (the sum-rate constraint on $\{0,1,\dots,m\}$ is implied by the sum-rate constraint on $\{0,1,\dots,i,i+2,\dots,m\}$ together with the individual constraint on user $i+1$). This is most easily understood by checking the equivalent condition for the deterministic channel. 
\end{remark}

This region \eqref{simpleSumRate} may be compared to the capacity region of a deterministic channel by making the correspondence $n_{ii}=\log\snr_i$, $0\leq i\leq K$, and $n_{0i}=\log\inr_i$, $1\leq i\leq K$. With this choice, \eqref{simpleSumRate} gives for each $\S\subseteq\{1,\dots,K\}$ such that a relabeling of the indices allows $\S=\{1,\dots,m\}$ with $n_{0m}-n_{mm}\leq n_{00}$,  $n_{0i}>0$, $0\leq i\leq K$, and also $n_{0i}-n_{ii}\leq n_{0,i+1}-n_{i+1,i+1}$ and $n_{0i}\leq n_{0,i+1}$ for $1\leq i\leq m-1$, the sum-rate constraint 
\begin{align}
    r_0+ r_1 +\dots +r_m
& \leq\label{e:MTOGauss_compareOuter1}
   \sum_{i=1}^m
   (n_{ii}-n_{0i})^+  
   +\sum_{i=1}^{m-1} \left(n_{0i}- (n_{0,i+1}-n_{i+1,i+1})^+\right)^+ 
\\ &\quad
   +\max(n_{0m},n_{00})
   +( m + 2 ) \log( m + 1 )\nonumber
\\&=n_{00}+\sum_{i=1}^m\left( (n_{0i}-n_{00})^++(n_{ii}-n_{0i})^+\right)\nonumber
\\&\quad +\sum_{k=1}^\an
   (|U_k\cap \S|-1)^+ + ( m + 2 ) \log( m + 1 ) \label{e:MTOGauss_compareOuter2}
\\&=f_{\text{free}}(\S)+\sum_{k=1}^{n_{00}}f_k(\S)+( m + 2 ) \log( m + 1 ) \,.
\end{align}
The step leading from \eqref{e:MTOGauss_compareOuter1} to \eqref{e:MTOGauss_compareOuter2} can be understood with the help of Figure~\ref{GaussianRectangles}. Each term in the second sum in \eqref{e:MTOGauss_compareOuter1} counts the overlap of rectangle $i$ with rectangle $i+1$. By the conditions \eqref{conditions} the signal of each user that interferes above user 0's signal (for user $i$ this is $(n_{0i}-n_{00})^+$ levels) also overlaps with the signal from user $m$, so is counted in this sum. Also, it is not hard to see that each level is counted exactly once fewer times than the number of users interfering at that level, giving rise to the term $\sum_{k=1}^\an (|U_k\cap \S|-1)^+$. 

Evidently, the Gaussian many-to-one IC outer bound lies within $3\log K$ bits per user of the corresponding deterministic channel.

All the ingredients are in place for the main result of the paper.

\begin{theorem}\label{t:smallGap}
The capacity region of the Gaussian many-to-one interference channel lies within $(2K+5)\log K$ bits per user of the region given in Lemma~\ref{sumRateLemma}.
\end{theorem}
\begin{proof}
  Directly comparing the outer bound with the achievable region would require proving a counterpart to Lemma~\ref{deterministicAchievable}. Fortunately, the outer bound and the achievable region have each already been compared to the capacity region of a corresponding deterministic channel, expressed in two different ways. This upper bounds the gap between the achievable region and outer bound, proving the theorem.
\end{proof}

The notion of the generalized degrees-of-freedom region, defined in \cite{OneBit}, gives insight towards the behavior at high $\snr$ and $\inr$. 
The generalized degrees-of-freedom region for the many-to-one IC is found by putting $\snr_i=s^{\alpha_i}$ for $0\leq i\leq K$ and $\inr_i=s^{\beta_i}$ and taking the limit $s\to\infty$. The constants $\alpha_i$ and $\beta_i$ are proportional to $\snr_i$ and $\inr_i$ in the dB scale. Let $C(s,\overrightarrow{\alpha},\overrightarrow{\beta})$ be the capacity region of a many-to-one IC with $\{\snr_i\},\{\inr_i\}$ thus defined. The resulting degrees-of-freedom region is $$D(\overrightarrow{\alpha},\overrightarrow{\beta})=\lim_{s\to\infty} \frac{C(s,\overrightarrow{\alpha},\overrightarrow{\beta})}{\log s}\, .$$  
To evaluate this limit, note that Theorem~\ref{t:smallGap} allows to directly calculate the degrees-of-freedom from the outer bound of Lemma~\ref{sumRateLemma}:
\begin{corollary}\label{ManyToOneDF}
The generalized degrees-of-freedom region of the Gaussian many-to-one channel is given by the set of points $(d_0,d_1,\dots,d_K)$ satisfying each of the individual constraints
$$d_i\leq \alpha_i,\quad 0\leq i\leq K\, ,$$
and for each $\S\subseteq \{1,\ldots,K\}$ with the property that a relabeling of the indices of $\S$ allows $ \S =\{1,\ldots,m\} $ (where $m=|\S|$) such that
\begin{equation}\label{conditions2}
\alpha_0>0,\quad \beta_m-\alpha_m \leq \alpha_0, 
\quad
\beta_i>0,\quad 1\leq i\leq m.
\end{equation}
and $\beta_i-\alpha_i\leq \beta_{i+1}-\alpha_{i+1}$ for $1\leq i\leq m-1$, the following sum-rate constraint holds:
\begin{align}
d_0+ d_1 +\dots +d_m \nonumber
& \leq
   \sum_{i=1}^m
   (\alpha_i-\beta_i)^+
   +\sum_{i=1}^{m-1} \left(\beta_i-\left(\beta_{i+1}-\alpha_{i+1}\right)^+\right)^+ \nonumber
   +\max(\beta_m,\alpha_0) .\nonumber
\end{align}

\end{corollary}

\begin{remark}
  This is exactly the scaled capacity region of a particular deterministic channel, assuming $\{\alpha_i\},\{\beta_i\}$ are rational numbers. The first sum accounts for the part of each signal that is received below the noise level at user 0. The second sum corresponds to the number of users minus one on levels with multiple interferers, and the final term is the rate that is achieved with each level used exactly once up to $\beta_m$ or $\alpha_0$, whichever is larger. The constraint may be compared to \eqref{sumRateConstraint}, recalling the conditions \eqref{conditions2}.  
\end{remark}

This concludes the treatment of the many-to-one IC. The second half of the paper tackles the one-to-many IC.

\section{Deterministic One-to-Many Interference Channel}\label{s:OneToManyDet}
Consider the channel obtained by reversing the roles of the transmitters and receivers in the deterministic many-to-one IC of Section~\ref{s:DeterministicManytoOneCapacity}. More precisely, if the original channel has gains $\nt_{ii}, 0\leq i\leq K$ and $\nt_{0i}, 1\leq i \leq K$, let the reversed channel have gains $ n_{ii}=\nt_{ii}, 0\leq i\leq K$ and $ n_{i0}=\nt_{0i}, 1\leq i \leq K$ (see Figure~\ref{ReverseZChannelDet}). 

Recall the simple capacity achieving scheme for the deterministic many-to-one IC: each level as observed at receiver 0 is allocated entirely to user 0 or to all users causing interference on the level. The corresponding achievable scheme for the deterministic one-to-many IC allocates each level as observed at \emph{transmitter} 0 either to user 0 or to all other users \emph{experiencing} interference from this level. A little thought reveals that the two achievable regions are the same, and one suspects that the capacity regions are the same as well. This is confirmed by the following theorem. Thus, the many-to-one and one-to-many channels are reciprocal (see \cite{TseViswanath:01} for a discussion of reciprocal channels). 
\begin{figure}
\begin{centering}
\psset{unit=.7mm,linewidth=.3pt,arrowlength=1.2,arrowinset=0,labelsep=2pt}
\begin{center}
\begin{pspicture}(-1,-6)(90,70)
\rput(56,0){
\rput(0,-12){

\pscircle(0,0){1.5}
\pscircle(0,4){1.5}
\pscircle(0,8){1.5}

\psline{->}(1.5,8)(28.5,8)
\psline{->}(1.5,4)(28.5,4)
\psline{->}(1.5,0)(28.5,0)
\pspolygon(-6,-4)(-6,12)(-2,12)(-2,-4)

\uput[l](-6,4){\small${\text{Tx}_3}$}
\uput[r](36,8){\small${\text{Rx}_3}$}}

\rput(0,16){
\pscircle(0,0){1.5}
\psline{->}(1.5,0)(28.5,0)
\pspolygon(-6,4)(-6,-4)(-2,-4)(-2,4)
}

\psline{->}(1.5,68)(28.5,20)
\psline{->}(1.5,64)(28.5,16)

\rput(0,24){
\pscircle(0,8){1.5}
\pscircle(0,12){1.5}
\pscircle(0,16){1.5}
\psline{->}(1.5,16)(28.5,16)
\psline{->}(1.5,8)(28.5,8)
\psline{->}(1.5,12)(28.5,12)
\pspolygon(-6,4)(-6,20)(-2,20)(-2,4)}

\uput[l](-6,34){\small$\text{Tx}_1$}
\uput[r](36,34){\small$\text{Rx}_1$}

\psline{->}(1.5,68)(28.5,4)
\psline{->}(1.5,64)(28.5,0)
\psline{->}(1.5,60)(28.5,-4)
\psline{->}(1.5,56)(28.5,-8)
\psline{->}(1.5,52)(28.5,-12)

\psline{->}(1.5,60)(28.5,32)
\psline{->}(1.5,64)(28.5,36)
\psline{->}(1.5,68)(28.5,40)

\rput(0,52){\pscircle(0,0){1.5}
\pscircle(0,4){1.5}
\pscircle(0,8){1.5}
\pscircle(0,12){1.5}
\pscircle(0,16){1.5}
\psline{->}(1.5,4)(28.5,4)
\psline{->}(1.5,16)(28.5,16)
\psline{->}(1.5,8)(28.5,8)
\psline{->}(1.5,12)(28.5,12)
\psline{->}(1.5,0)(28.5,0)
\pspolygon(-6,-4)(-6,20)(-2,20)(-2,-4)
}

\uput[l](-6,60){\small$ \text{Tx}_0$}
\uput[r](36,60){\small$ \text{Rx}_0$}

\rput(30,-12){
\pscircle(0,0){1.5}
\pscircle(0,4){1.5}
\pscircle(0,8){1.5}
\pscircle(0,12){1.5}
\pscircle(0,16){1.5}
\pspolygon(6,-4)(6,20)(2,20)(2,-4)}

\rput(30,16){
\pscircle(0,0){1.5}
\pscircle(0,4){1.5}
\pspolygon(6,8)(6,-4)(2,-4)(2,8)
}
\uput[l](-6,16){\small$ \text{Tx}_2$}
\uput[r](36,18){\small$ \text{Rx}_2$}

\rput(30,24){
\pscircle(0,8){1.5}
\pscircle(0,12){1.5}
\pscircle(0,16){1.5}
\pspolygon(6,4)(6,20)(2,20)(2,4)}

\rput(30,52){\pscircle(0,0){1.5}
\pscircle(0,4){1.5}
\pscircle(0,8){1.5}
\pscircle(0,12){1.5}
\pscircle(0,16){1.5}
\pspolygon(6,-4)(6,20)(2,20)(2,-4)
}}


\rput(-16,-6){
\psline[arrowlength=1,linewidth=.4]{->}(5,0)(26,0)
\psline[arrowlength=1,linewidth=.4]{->}(5,23)(26,23)
\psline[arrowlength=1,linewidth=.4]{->}(5,46)(26,46)
\psline[arrowlength=1,linewidth=.4]{->}(5,69)(34,69)


\psline[arrowlength=1,linewidth=.4]{->}(30,69)(34,69)

\psline[arrowlength=1,linewidth=.4]{->}(4,69)(27.2,2)
\psline[arrowlength=1,linewidth=.4]{->}(4,69)(27.2,25)
\psline[arrowlength=1,linewidth=.4]{->}(4,69)(27.2,48)

\psline[arrowlength=1,linewidth=.4]{->}(30,46)(34,46)
\psline[arrowlength=1,linewidth=.4]{->}(30,23)(34,23)
\psline[arrowlength=1,linewidth=.4]{->}(30,0)(34,0)

\psline{-}(27,46)(29,46)
\psline{-}(28,45)(28,47)
\pscircle(28,46){2}
\psline{-}(27,23)(29,23)
\psline{-}(28,22)(28,24)
\pscircle(28,23){2}
\psline{-}(27,0)(29,0)
\psline{-}(28,-1)(28,1)
\pscircle(28,0){2}

\pscircle(35,69){1}
\pscircle(4,0){1}
\pscircle(35,46){1}
\pscircle(4,46){1}
\pscircle(4,23){1}
\pscircle[fillstyle=solid,fillcolor=white](4,69){1}
\pscircle(35,23){1}
\pscircle(35,0){1}

\uput{2pt}[l](24,10){\small 5}
\uput{2pt}[u](25,31){\small 2}
\uput{2pt}[u](22,54){\small 3}
\uput{2pt}[u](11,69){\small 5}
\uput{2pt}[u](11,23){\small 1}
\uput{2pt}[u](8,46){\small 3}
\uput{2pt}[u](11,0){\small 3}
}
\end{pspicture}
\end{center} 
\caption{The one-to-many interference channel in this figure is obtained by reversing the roles of transmitters and receivers in the many-to-one channel in Figure~\ref{DeterministicZChannelFigure}.}\label{ReverseZChannelDet}
\end{centering}
\end{figure}
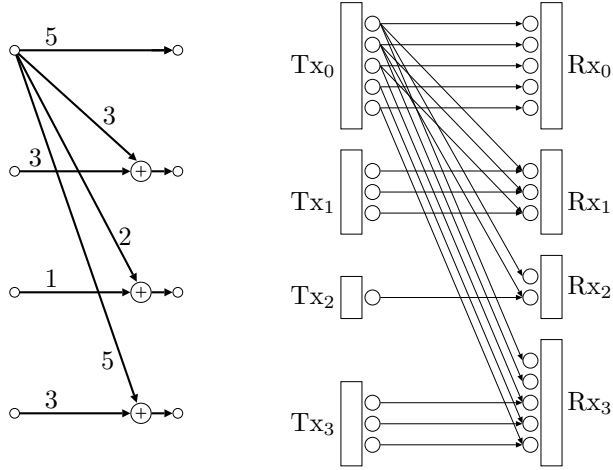

\begin{theorem}\label{ReversedDeterministicCapacity}
The capacity region of a deterministic one-to-many IC with channel gains $ n_{ii}, 0\leq i\leq K$ and $ n_{0i}, 1\leq i \leq K$, is equal to the capacity region of a deterministic many-to-one IC (as given in Lemma~\ref{l:outerBound}) with gains $ \tilde n_{ii}= n_{ii}, 0\leq i\leq K$ and $\tilde n_{0i}= n_{i0}, 1\leq i \leq K$.
\end{theorem}

The notation in this section is very similar to that used for the many-to-one deterministic interference channel of Section~\ref{s:DeterministicManytoOneCapacity}. Assume without loss of generality that $x_0 $ is restricted to the elements that appear in the output $y_0 $, i.e. $x_0\in \F_2^{n_{00}}$. Denote by $U_k \subseteq \{1,\ldots,K\}$, $1\leq k\leq n_{00}$, the set of users
potentially experiencing interference from the $k$th level at transmitter 0: $U_k=\{i:1\leq i\leq K,  n_{00}- n_{i0}< k\leq n_{ii}-n_{i0}+n_{00}\} $.
For a user $i,1\leq i\leq K$, and a level $k,1\leq k\leq n_{00}$, denote by
$x_{i|k}$ the signal of user $i$, restricted to the level that overlaps with level $k$ of user 0's signal. Finally, let  
$\tilde x_i$ be the restriction of the input from transmitter $i$ to the lowest $(n_{i0}-n_{00})^+$ levels. This is the part of $x_i $ that appears below the interference from user 0. Similarly, let $\hat x_i $ be the restriction of the input from transmitter $i $ to the highest $ (n_{ii} -n_{i0 })^+ $ levels. This is the part of $x_i $ that lies above the interference from user 0. 

Let us quickly relate the sets $U_k$ for the one-to-many channel to the analogous sets in the many-to-one channel. As in Theorem~\ref{ReversedDeterministicCapacity}, consider a many-to-one channel with gains $\tilde n_{ii}=n_{ii},0\leq i\leq K$ and $\tilde n_{0i}=n_{i0},1\leq i\leq K$. Denote the set of users experiencing interference from the $k$th level of user 0 by $\widetilde U_k=\{i:1\leq i\leq K, \tilde n_{0i}-\tilde n_{ii}<k\leq \tilde n_{0i}\}$ (see Section~\ref{s:DeterministicManytoOneCapacity}). 
It holds that  
\begin{align*}
 &\tilde n_{0i}-\tilde n_{ii}<k\leq \tilde n_{0i}
 \\ &\Leftrightarrow n_{i0}- n_{ii}<k\leq n_{i0}
  \\&\Leftrightarrow n_{i0}- n_{ii}-n_{00}<k-n_{00}\leq n_{i0}-n_{00}
  \\&\Leftrightarrow n_{00}- n_{i0}< 1+n_{00}- k \leq n_{ii}-n_{i0}+n_{00}\,,
\end{align*} 
whence $\widetilde U_k=U_{1+n_{00}-k}$.
In particular, for any $\S\subseteq \{1,\dots, K\}$ we have
\begin{equation}\label{e:setsU_k}
  \sum_{k=1}^{n_{00}}|U_k\cap \S|=\sum_{k=1}^{n_{00}}|\widetilde U_k\cap \S|\,.
\end{equation}

Using this last equation we may explicitly write the capacity region of the deterministic many-to-one channel from Theorem~\ref{ReversedDeterministicCapacity} as those rate points satisfying the individual rate constraints
\begin{equation}
   r_i\leq {\al_ {ii}}, \quad 0\leq i\leq K,
\end{equation} and the $2^K-1$ sum-rate constraints, one for each non-empty $\S\subseteq \{1,\ldots,K\}$,
\begin{equation}\begin{split}\label{e:OTMdet_sumrateConstraint}
r_0+\sum_{i\in\S}r_i \leq n_{00}+\sum_{k=1}^\an
   (|U_k\cap \S|-1)^+     +\left(\sum_{i\in \S} (n_{i0}-n_{ii})^+
  +(n_{i0}-n_{00})^+\right)\,.
\end{split}\end{equation}

\subsection{Proof of Outer Bound}
We may (without loss of generality) order the users so that $n_{i0}\leq n_{i+1,0},1\leq i\leq K-1$. As before, the rate across each link cannot exceed the point-to-point capacity, hence
\begin{equation}r_i\leq n_ {ii}, \quad 0\leq i\leq K.\end{equation}

Next, we prove the claimed sum-rate constraint on a set of users $\S\cup \{0\}$, where $\S\subseteq \{1,\ldots,K\}$. Unlike the deterministic many-to-one channel, no side information is required to prove the constraint. For each $1\leq i\leq K$, let $\sig_i=\{x_{0|k}:(n_{00}-n_{i0}+n_{ii})^+<k\leq n_{00}\}$ be the part of signal 0 that appears above the intended signal at receiver $i$. Note that by the definition $\sig_i$ is determined by $y_i$, and also $\sig_i$ is independent of $x_i$, hence
\begin{align*}
  I(x_i^N;y_i^N)&=H(x_i^N)-H(x_i^N|y_i^N)
  \\&=H(x_i^N|\sig_i^N)-H(x_i^N|\sig_i^N,y_i^N)
  \\&=I(x_i^N;y_i^N|\sig_i^N)\,.
\end{align*}
Now, Fano's inequality and the data processing inequality give
\begin{align*}
 N(r_0 +\sum_{i\in \S} r_i-\eps_N)\nonumber
&\leq
    I(x_0^N; y_0^N)+\sum_{i\in \S} I(y_i^N;x_i^N)
\\&=
I(x_0^N; y_0^N)+\sum_{i\in \S} I(y_i^N;x_i^N|\sig_i^N)
\\&=
    H(x_0^N)+ \sum_{i\in \S} \big(H(y_i^N|\sig_i^N)-H(y_i^N|x_i^N,\sig_i^N)\big)\, .
\end{align*}
Breaking the signals apart by level, using the independence bound on entropy, the chain rule for entropy, and removing conditioning, we may rewrite the above as
\begin{align*}
\\&=
    \sum_{i\in \S} \bigg(H\left(\tilde x_i^N,\{x^N_{i|k}+x^N_{0,k}:k \text{ s.t. }i\in U_k\},\hat x^N_i\right)
    -H\left( \{x^N_{0,k}:k \text{ s.t. }i\in U_k\} \right)\bigg) +H\left( \{x^N_{0|k}\}_{k=1}^\an \right)
\\&\leq 
     \sum_{i\in \S} \bigg(H(\tilde x^N_i) + H(\hat x^N_i)+H\left(\{x^N_{i|k}+x^N_{0,k}:k \text{ s.t. }i\in U_k\}\right)\bigg)+
    H\left(\{x_{0|k}^N: k \text{ s.t. } \S\cap U_k=\varnothing\}\right)
\\&\leq
    N\bigg(\sum_{i\in \S} \big((n_{i0}-n_{00})^+ +(n_{ii} -n_{i0 })^+ \big)
+\sum_{k=1}^{n_{00}} \max(|U_k\cap \S|,1)\bigg)\, .
\end{align*}
Taking $N\to \infty$ proves the constraint.
\hfill \qedsymbol

\subsection{Achievability of Outer Bound}
As mentioned before, the achievable scheme is nearly the same as that of the deterministic many-to-one IC, with either user 0 or all other users transmitting on a level. Each level $1\leq k\leq n_{00}$ viewed individually has capacity $C_k$, where $C_k$ is given by \eqref{e:SingleLevelIndividual} and \eqref{e:SingleLevelPairWise}. By transmitting on levels above and below the interference from user 0, the region $C_\text{free}$ is achievable without affecting the remaining levels. Thus, the achievable region $$\CalC_{\text{free}}+\sum_{k=1}^\an \CalC_k$$ is exactly the same as for the deterministic many-to-one channel \eqref{e:AchievableRegionParallel}. Also, the outer bound is the same as for the many-to-one channel, and since they match by Theorem~\ref{t:DeterministicCapacity}, this completes the proof of Theorem~\ref{ReversedDeterministicCapacity}. \hfill\qedsymbol

\subsection{Generalized Han-Kobayashi Scheme}
The achievable scheme of the previous section treats each level separately. In the Gaussian one-to-many IC, however, instead of decomposing the channel into independent sub-channels by level, it will turn out to be more natural to consider a generalized Han-Kobayashi (HK) scheme. Comparing the achievable region of the Han-Kobayashi scheme to the outer bound is most readily performed in the deterministic setting, where the two regions are equal. Therefore, we give a HK scheme for the deterministic channel.
 
Assume without loss of generality that the users are ordered by increasing interference from user 0, i.e. $n_{i0}\geq n_{i-1,0}$ for $2\leq i\leq K$, and that $n_{10}\geq 1$ and $n_{K0}-n_{00}\leq n_{KK}$ (so that all users actually experience interference from user 0).  To simplify the subsequent definitions we put $n'_{00}=0$ and $n'_{i0}=n_{i0}$ for $1\leq i\leq K$. Note that the truncation of signal $0$ at receiver $i$ occurs at level $(n_{00}-n_{i0})^+$, i.e. this is the highest level that is truncated.
With this in mind, the signal from user 0 decomposes naturally according to which users can observe each level:
let the $i$th signal, $1\leq i\leq K+1$, from user 0 be $$X_0(i)=\{x_{0|k}:(n_{00}-n_{i0})^+< k\leq n_{00}-n'_{i-1,0}\}\,,$$
and let $$X_0(K+1)=\{x_{0|k}:1\leq  k\leq (n_{00}-n_{K0})^+ \}\,.$$
The signals $X_0(1),\dots,X_0(i)$ are received by user $i$ above the noise level and are decoded, i.e. signal $X_0(i)$ is common information to users $i,\dots,K$. The signal $X_0(K+1)$ (possibly vacuous in the case $n_{K0}\geq n_{00}$) is received  below the noise level of all users except user 0, and is therefore private information.

Each user $i$, $1\leq i\leq K$, jointly decodes the intended signal $x_i$ together with $X_0(1),\dots,X_0(i)$. Thus, the achievable rate region is given by the intersection of a collection of multiple access channels, one at each receiver. Denote the rate of signal $X_0(i)$ by $R_0(i)$. The MAC constraints at receiver 0 (on the rates $R_0(1),\dots,R_0(K+1)$) are implied by the ``individual" rate constraints 
\begin{equation}
  \begin{split}
    \label{e:OTMdet_usr0rateConstraints}
    R_0(k)&\leq n_{k0}-n_{k-1,0},\quad  2\leq k\leq K\,, \\ R_0(1)&\leq \min(n_{10},n_{00})\,, \\
    R_0(K+1)&\leq (n_{00}-n_{K0})^+\,.
  \end{split}
\end{equation}

Some notation is necessary to cleanly express the constraints at the other receivers. Let $\lambda(i)\in\{1,\dots,i\}$ be such that the signal $X_0(\lambda(i))$ interferes at the top level of $x_i$ at receiver $i$, i.e. 
$$
n'_{\lambda(i),0}-n_{00}+n_{i0}<n_{ii} \leq n_{\lambda(i)-1,0}-n_{00}+n_{i0}\,.
$$
Rearranging, we have
$$n'_{\lambda(i),0}< n_{ii}+n_{00}-n_{i0}\leq n_{\lambda(i)-1,0}$$ if there is such a $\lambda$, and otherwise set $\lambda(i)=0$. 
For example, in Figure~\ref{reciprocalRectanglesFig} we have $\lambda(1)=0$, 
$\lambda(2)= 1$, and $\lambda(3)=3$.

Now, the signals $X_0(1),\dots,X_0(\lambda(i)-1)$ appear above the signal $x_i$, and are therefore observed cleanly. Thus, the MAC constraints at receiver $i$, $1\leq i\leq K$ are implied by the following subset of constraints: the above individual constraints \eqref{e:OTMdet_usr0rateConstraints} on $R_0(1),\dots,R_0(i)$ and the individual constraint 
\begin{equation}\label{e:OTMdet_indconstraint}r_i\leq n_{ii}\,,\end{equation}
 together with the sum-rate constraints
\begin{equation}\label{e:OneToManyDetHK}
  r_i+\sum_{k=\lambda(i)}^{i} R_0(k)\leq n_{i0}-n'_{\lambda(i)-1,0}
\end{equation}
and
\begin{equation}\label{e:OneToManyDetHK2}
  r_i+\sum_{k=\lambda(i)+1}^{i} R_0(k)\leq n_{ii}\,.
\end{equation}

We now check that the achievable region for the Han-Kobayashi scheme contains the achievable region from the previous section obtained by considering each level separately. First, the sum-rate constraints \eqref{e:OneToManyDetHK} and \eqref{e:OneToManyDetHK2} at each user are easily seen to result from adding the pairwise constraints \eqref{e:SingleLevelPairWise} on users $i$ and $0$ on the relevant levels. Similarly, the individual constraints on the rates $R_0(1),\dots,R_0(K+1)$ are implied by adding the individual constraints \eqref{e:SingleLevelIndividual} on the relevant levels. Thus, the constraints defining the Han-Kobayashi achievable region are looser than those defining the capacity-achieving scheme, and hence the Han-Kobayashi scheme achieves capacity. These conclusions are recorded in the following proposition.
\begin{proposition}\label{p:OTMdetHKregion}
  The capacity region of the deterministic one-to-many IC is achieved using a generalized Han-Kobayashi scheme as described above and can be expressed by the constraints \eqref{e:OTMdet_usr0rateConstraints}, \eqref{e:OTMdet_indconstraint}, \eqref{e:OneToManyDetHK}, and \eqref{e:OneToManyDetHK2}.
\end{proposition}

\section{Approximate Capacity Region of the One-to-Many Gaussian Interference Channel}\label{s:OneToManyGaussian}
Define the signal to noise ratios $\snr_i=|h_{ii}|^2 P_i/\n,0\leq i\leq K$ and $\inr_i=|h_{i0}|^2 P_0/\n,1\leq i\leq K$.
We assume the users are ordered by increasing values of $\inr$, i.e. $\inr_{i+1}>\inr_i$ for $1\leq i\leq K-1$. Moreover, we assume as in the many-to-one IC that $\inr_1> 1$: any user with $\inr_i\leq 1$ can simply treat the interference as noise and lose at most 1 bit relative to the point to point AWGN channel. Figure~\ref{ReverseZChannel} depicts the one-to-many Gaussian interference channel.

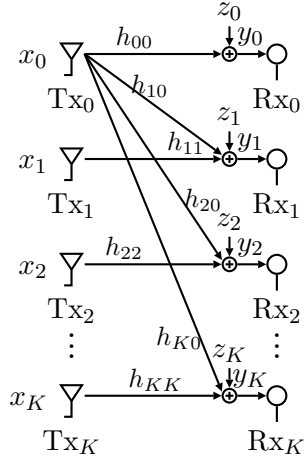
\begin{figure}
\begin{centering}
\psset{unit=.7mm,arrowlength=1,arrowinset=0}
\begin{center}
\begin{pspicture}(-18,-35)(20,42)
    \rput(-16,35){%
            \psline(0,0)(2,3)(-2,3)(0,0)
            \psline(0,0)(0,-3)(-1.4,-3)
            \psline{->}(2.5,1)(28.5,1)
            \pscircle(30,1){1.5}
            \psline{-}(30,.25)(30,1.75) \psline{-}(29.25,1)(30.75,1)
            \psline{->}(31.5,1)(37,1)
            \pscircle(39,1){2}
            \psline{-}(39,-1)(39,-4)
            \uput[d](0,-3){$\text{Tx}_0$}
            \uput[l](-2,0){$x_0$}
            \uput[d](39,-3){$\text{Rx}_0$}
            \uput{2pt}[u](12,1){\small $h_{00}$}
            \uput{3pt}[u](34,1){$y_0$}
            \psline{->}(30,6.25)(30,2.5)
            \uput[u](30,5){$z_0$}
            }

\psline{->}(-13.5,36)(12.6, 16.8)
\psline{->}(-13.5,36)(12.6, -3.2)
\psline{->}(-13.5,36)(12.6, -28.2)

    \rput(-1.2,30.0){\small $h_{10}$}

    \rput(-16,15){%
            \psline(0,0)(2,3)(-2,3)(0,0)
            \psline(0,0)(0,-3)(-1.4,-3)
            \psline{->}(2.5,1)(28.5,1)
            \pscircle(30,1){1.5}
            \psline{-}(30,.25)(30,1.75) \psline{-}(29.25,1)(30.75,1)
            \psline{->}(31.5,1)(37,1)
            \pscircle(39,1){2}
            \psline{-}(39,-1)(39,-4)
            \uput[d](0,-3){$\text{Tx}_1$}
            \uput[l](-2,0){$x_1$}
            \uput[d](39,-3){$\text{Rx}_1$}
            \uput{2pt}[u](21.7,1){\small $h_{11}$}
            \uput{3pt}[u](34,1){$y_1$}
            \psline{->}(30,6.25)(30,2.5)
            \uput[u](30,5){$z_1$}
            }

    \rput(-16,-5){%
            \psline(0,0)(2,3)(-2,3)(0,0)
            \psline(0,0)(0,-3)(-1.4,-3)
            \psline{->}(2.5,1)(28.5,1)
            \pscircle(30,1){1.5}
            \psline{-}(30,.25)(30,1.75) \psline{-}(29.25,1)(30.75,1)
            \psline{->}(31.5,1)(37,1)
            \pscircle(39,1){2}
            \psline{-}(39,-1)(39,-4)
            \uput[d](0,-3){$\text{Tx}_2$}
            \uput[l](-2,0){$x_2$}
            \uput[d](39,-3){$\text{Rx}_2$}
            \uput{2pt}[u](10,1){\small $h_{22}$}
            \uput{3pt}[u](34,1){$y_2$}
            \psline{->}(30,6.25)(30,2.5)
            \uput[u](30,5){$z_2$}
            }

    \rput(9,8.2){\small $h_{20}$}

    \rput(-16,-18){\Large{$\vdots$}}
    \rput(23,-18){\Large{$\vdots$}}

    \rput(-16,-30){%
            \psline(0,0)(2,3)(-2,3)(0,0)
            \psline(0,0)(0,-3)(-1.4,-3)
            \psline{->}(2.5,1)(28.5,1)
            \pscircle(30,1){1.5}
            \psline{-}(30,.25)(30,1.75) \psline{-}(29.25,1)(30.75,1)
            \psline{->}(31.5,1)(37,1)
            \pscircle(39,1){2}
            \psline{-}(39,-1)(39,-4)
            \uput[d](0,-3){$\text{Tx}_K$}
            \uput[l](-2,0){$x_K$}
            \uput[d](39,-3){$\text{Rx}_K$}
            \uput{2pt}[u](16,1){\small $h_{KK}$}
            \uput{3pt}[u](34,1){$y_K$}
            \psline{->}(30,6.25)(30,2.5)
            \uput[u](30,5){$z_K$}
            }

    \rput(4.5,-18){\small $h_{K0}$}

\end{pspicture}
\end{center}
\caption{The one-to-many Gaussian interference channel has one user causing interference to $K$ other users.}\label{ReverseZChannel}
\end{centering}
\end{figure}

\begin{theorem} \label{t:OTMGauss_Capacity}
The capacity region of the one-to-many Gaussian IC with power-to-noise ratios $\snr_i$, $0\leq i\leq K$, and $\inr_i$, $1\leq i\leq K$, has capacity region within $(2K+1,1,\dots,1)$ bits of the region defined by the individual rate constraints \begin{equation}\label{e:OTMGauss_Ind}r_i\leq \log(1+ \snr_i),\quad 0\leq i\leq K\,,\end{equation}
and for each subset of users $\S\subseteq\{1,\dots,K\}$ with $|\S|=m$ relabeled as $\S=\{1,\dots,m\}$ such that $\inr_{i+1}>\inr_i$ for $1\leq i\leq K-1$ and $\inr_1> 1$, the sum-rate constraint
\begin{equation}\begin{split} 
\label{e:OTMGauss_SumRate}
r_0+\sum_{i=1}^m r_i&\leq 
    \log\left(1+\frac{\snr_0}{1+\inr_m}\right)+ \log(1+\snr_1+\inr_1)
\\&\quad + 
    \sum_{i=2}^m \log\left(1+\snr_i+\frac{\inr_i}{1+\inr_{i-1}}\right).
\end{split}\end{equation}
\end{theorem}

\subsection{Outer Bound}

In contrast with the deterministic case, side information is required to prove the sum-rate constraint. Let $\S\subseteq \{1,\dots,K\}$, and by relabeling, assume $\S=\{1,\dots,m\}$ where $m=|\S|$. Furthermore, assume $\inr_m-\snr_0\leq\snr_m$; otherwise receiver $m$ can cleanly decode the interference from user 0 while treating its own signal as noise, and the constraint is redundant. 

We give as side information to receiver 0 the interfering signal as observed at receiver $m$ (the receiver experiencing the greatest interference), and we give as side information to each receiver $i$, $2\leq i\leq m$, the interfering signal $x_0$ as observed at receiver $i-1$:
\begin{equation}
\begin{split}
s_0&=h_{m0}x_0+z_m
\\
    s_1&= \varnothing
\\
    s_i&= h_{i-1,0}x_0 + z_{i-1},\quad 2\leq i\leq m.
\end{split}
\end{equation}
Now, Fano's inequality and the Data Processing Inequality give
\begin{align}
N(r_0+\sum_{i=1}^m r_i-\eps_N) &\leq \nonumber
    I(y_0^N,s_0^N;x_0^N)+\sum_{i=1}^m I(y_i^N,s_i^N;x_0^N)
\\&= \nonumber
    h(y_0^N|s_0^N)+h(s_0^N)-h(z_0,z_m)
    \\&\quad    \label{GaussianReciprocal1}
        + \sum_{i=1}^m \left(h(y_i^N|s_i^N)+h(s_i^N)-h(y_i^N,s_i^N|x_i^N)\right).
\end{align}

The fact that conditioning reduces entropy implies that for $1\leq i\leq m-1$ 
\begin{align*}h(y_i^N,s_i^N|x_i)&=h(h_{i0}x_0^N+z_i^N,s_i^N)\\ &\geq h(h_{i0}x_0^N+z_i^N)= h(s_{i+1}^N),\end{align*}
and 
$$h(y_m^N,s_m^N|x_m^N)\geq h(y_m^N|x_m^N)=h(s_0).$$
Plugging this into equation (\ref{GaussianReciprocal1}), the sum telescopes, producing
\begin{align}
N(r_0+\sum_{i=1}^m r_i-\eps_N) \leq \nonumber
    h(y_0^N|s_0^N)  - h(z_0^N,z_m^N)+h(y_1)+
    \sum_{i=2}^{m} h(y_i^N|s_i^N).
\end{align}
We bound each term using the fact that the Gaussian distribution maximizes entropy for a fixed conditional variance:
\begin{align*}
h(y_0^N|s_0^N)
&=
    h(h_{00}x_0^N+z_0^N|h_{m0}x_0^N+z_m^N)
\\&\leq
     N\log\left(1+\frac{\snr_0}{1+\inr_m}\right)+N \log(\pi e \n),
\end{align*}
and for $2\leq i\leq m$,
\begin{align*}
h(y_i^N|s_i^N)
&=
    h(h_{i0}x_0^N+h_{ii}x_i^N+z_i^N|h_{i-1,0}x_0^N + z_{i-1}^N)
\\&\leq
      N\log\left(1+\snr_i+\frac{\inr_i}{1+\inr_{i-1}}\right)+N \log(\pi e \n).
\end{align*}
Also $$ h(y_1^N )\leq N\log(1+\snr_1+\inr_1)+N\log(\pi e \n).$$
Combining these calculations and taking $N\to \infty$, we have the sum-rate constraint:
\begin{equation}\begin{split}\label{e:oneToManySumRateConstraint}
r_0+\sum_{i=1}^m r_i
&\leq 
    \log\left(1+\frac{\snr_0}{1+\inr_m}\right)+ \log(1+\snr_1+\inr_1)
\\&\quad +
    \sum_{i=2}^m \log\left(1+\snr_i+\frac{\inr_i}{1+\inr_{i-1}}\right).
\end{split}\end{equation}
\hfill\qedsymbol

Before proceeding with the achievable scheme, let us first rewrite this region in a form that will allow to easily compare with the deterministic channel region \eqref{e:OTMdet_usr0rateConstraints}, \eqref{e:OTMdet_indconstraint}, \eqref{e:OneToManyDetHK}, \eqref{e:OneToManyDetHK2}.
Let $n_{ii}=\snr_i$, $0\leq i\leq K$, and $n_{i0}=\inr_i$, $1\leq i\leq K$.
The region given by \eqref{e:OTMGauss_Ind} and \eqref{e:OTMGauss_SumRate} may be enlarged to give the region defined by the constraints
\begin{equation}
  r_i\leq 1+n_{ii},\quad 0\leq i\leq K\,,
\end{equation}
and for each subset of users $\S\subseteq \{1,\dots,K\}$ as above,
\begin{equation}
  \begin{split}
    r_0+\sum_{i=1}^m r_i
    &\leq (m+1)+
     (n_{00}-n_{m0})+\max(n_{11},n_{10})\\&\quad+\sum_{i=2}^m\max(n_{ii},n_{i0}-n_{i-1,0})\,.
  \end{split}
\end{equation}
Summing over levels instead of users yields exactly the sum-rate constraint \eqref{e:OTMdet_sumrateConstraint} of the deterministic channel with an added gap of $m+1$. Viewing the gap as coming entirely from the rate of user 0, the deterministic achievable region is within $K+1$ bits/s/Hz at user 0 of the outer bound.

\begin{remark}The constraint can be interpreted using Figure~\ref{reciprocalRectanglesFig}. In the figure, the received signal power occupancy at each receiver is superimposed in the appropriate position relative to the signal of user 0. The noise floor of user $i$ is $\log \inr_i$ levels from the top of user 0's signal. As in the deterministic channel, on each level either user 0 transmits or all other users transmit. The sum-rate constraint \eqref{e:oneToManySumRateConstraint} counts each level once if user 0 causes interference to one or no users, and equal to the number of users if more than two users are interfered by a level. 
\end{remark}
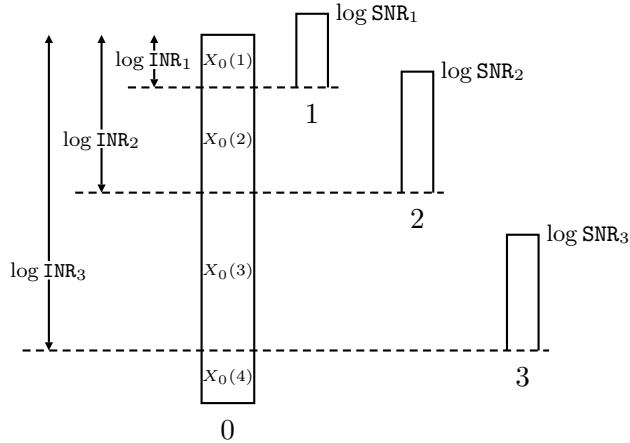
\begin{figure}
\begin{centering}
\psset{unit=.7mm,arrowlength=1,arrowinset=0,labelsep=3pt}
\begin{center}
\begin{pspicture}(-28,-10)(74,70)
\pspolygon(6,-10)(6,60)(16,60)(16,-10) \rput(11,-15){0}
\psline(24,50)(24,64)(30,64)(30,50)\rput(27,45){1}
\psline(44,30)(44,53)(50,53)(50,30)\rput(47,25){2}
\psline(64,0)(64,22)(70,22)(70,0)\rput(67,-5){3}

\psline[linestyle=dashed,dash=3pt 2pt](-8,50)(32,50)
\psline[linestyle=dashed,dash=3pt 2pt](-18,30)(52,30)
\psline[linestyle=dashed,dash=3pt 2pt](-28,0)(72,0)

\rput(11,55){\tiny $X_0(1)$}
\rput(11,40){\tiny $X_0(2)$}
\rput(11,15){\tiny $X_0(3)$}
\rput(11,-5){\tiny $X_0(4)$}

\rput(-23,15){\scriptsize $\log \inr_3$}
\psline{->}(-23,17)(-23,60)
\psline{->}(-23,13.5)(-23,0)

\rput(-13,40){\scriptsize $\log \inr_2$}
\psline{->}(-13,42)(-13,60)
\psline{->}(-13,38.5)(-13,30)

\rput(-3,55){\scriptsize $\log \inr_1$}
\psline{->}(-3,57)(-3,60)
\psline{->}(-3,53.5)(-3,50)

\uput[r](30,64){\footnotesize $\log \snr_1$} 
\uput[r](50,53){\footnotesize $\log \snr_2$} 
\uput[r](70,22){\footnotesize $\log \snr_3$}

\end{pspicture}
\end{center} 
\caption{A superposition of the received signal levels at each user for an example channel. The dashed lines indicate the noise floor for each receiver. User 0 employs a superposition code with codebooks $\{X_0(k)\}_{k=1}^i$ intended for receiver $i$. }
\label{reciprocalRectanglesFig}
\end{centering}
\end{figure}

\subsection{Achievable Region}

As in the many-to-one channel, the achievable strategy emulates the approach used for the deterministic case. In the many-to-one channel lattice codes are used to align the interference at receiver 0; in contrast, since there are only two signals at each receiver in the one-to-many channel, it suffices to adopt a rate-splitting approach using a superposition of random Gaussian codebooks.  The scheme is completely analogous to the Han-Kobayashi scheme for the deterministic channel of the previous section.


We now describe the achievable scheme. In constructing the scheme we temporarily assume that $\inr_K\leq \snr_K$.
Transmitter 0 uses a superposition of independent Gaussian codebooks (see Figure~\ref{reciprocalRectanglesFig}), 
\begin{equation}\label{e:GaussianSuperpositionCode}
X_0=\sum_{k=1}^{K+1} X_0(k).
 \end{equation}
Each codebook corresponds to a level; the power range $[\frac{P_0}{\snr_0},P_0]$ used by transmitter 0 is divided into intervals (levels) according to the interference caused, just as in the deterministic case. More precisely, the values $P_0/\inr_i$ for $1\leq i\leq K$ partition the interval $[\frac{P_0}{\snr_0},P_0]$ into power levels. The power and rate associated with each level is that which would be assigned to user 0 using rate-splitting, with a small reduction in rate to be described later (see, e.g., \cite{RateSplittingMAC} for more detail on rate-splitting). 
Each user $i$, $1\leq i\leq K$, uses a random Gaussian codebook at full power, i.e. received at power $\snr_i$ by receiver $i$. Receiver $i$ first decodes those codebooks from user 0 that are received above the intended signal (while treating all other signals as noise), and then jointly decodes the signal from transmitter $i$ and the remaining signals from user 0 which are received above the noise level (treating interference received below the noise level as noise). 


Recall the assumptions $\snr_0>1$, $\inr_1>1$ and $\inr_{i+1}\geq \inr_{i}$.
The power of each of user 0's codebooks in \eqref{e:GaussianSuperpositionCode} is chosen in such a way that the sum of codebooks $i+1$ through $K+1$, $\sum_{k=i+1}^{K+1}X_0(k)$, is observed by receiver $i$ to be at the noise level (assuming all the codebooks are used). More precisely, letting $q_i$ denote the power transmitted in codebooks $i+1$ through $K+1$, $$q_i=\sum_{k=i+1}^{K+1}|X_0(k)|^2\,.$$ We require \begin{equation*}
  q_{i} |h_{i,0}|^2=N_0\, ,
\end{equation*}
or equivalently, \begin{equation}\label{e:OneToManyPowerLevel}q_{i}=\frac{P_0}{\inr_i},\quad 1\leq i\leq K\, .\end{equation}
The $i$th power interval is given by $[q_{i},q_{i-1}], 1\leq i\leq K+1$, where $q_0=P_0$, $q_{K+1}=\frac{P_0}{\snr_0}$, and $q_{i}=\frac{P_0}{\inr_i}$ for $1\leq i\leq K$. 
The power used by transmitter $0$ on level $i,1\leq i\leq K$ is $$\theta_i=\E[|X_0(i)|^2]=q_{i-1}-q_i\, ,$$ so that user 0 satisfies the power constraint (assuming user 0 transmits on all levels): $$\E[|X_0|^2]=\sum_{i=1}^{K+1} \E[|X_0(i)|^2]=\sum_{i=1}^{K+1} \theta_i=q_0-q_{K+1}\leq P_0\, .$$ 
Receiver 0 decodes the signals sequentially from the highest level (lowest index) downwards, treating the weaker signals as noise and subtracting off the decoded signal at each step. 
Thus, when decoding level $i$ receiver 0 experiences an effective noise variance of at most 
\begin{align*}N_0(i)\leq N_0+ |h_{00}|^2 \sum_{k>i}\theta_k  = |h_{00}|^2 q_{i} = N_0\frac{\snr_0}{\inr_i}\, .\end{align*} 

The rates of user 0's codebooks are chosen to satisfy the inequalities
\begin{equation}\label{e:OneToManyG_user0rates}
  R_0(i)\leq \log\left(1+ \frac{\theta_i|h_{i0}|^2}{3N_0}\right),\quad 1\leq i\leq K+1\,.
\end{equation}
Note that user 0 can decode its own signals since
\begin{align*}
  R_0(i)\leq \log\left(1+ \frac{\theta_i|h_{i0}|^2}{3N_0}\right)= 
  \log\left(1+ \frac{\theta_i|h_{00}|^2}{3N_0\frac{\snr_0}{\inr_i}}\right)
  \leq 
  \log\left(1+ \frac{\theta_i|h_{00}|^2}{N_0(i)}\right)
  \,.
\end{align*}  The quantity $\theta_i|h_{00}|^2/N_0(i)$ is the SINR of the $i$th signal from user 0. 

We now account for decoding at receiver $i$, $1\leq i\leq K$. A natural procedure is for receiver $i$ to jointly decode the $i$ strongest levels from user 0, i.e. $X_0(1),\dots,X_0(i)$, along with its own signal $X_i$. Since Gaussian codebooks are used, which is optimal for the MAC, it follows that the achievable region is determined by the MAC region at each receiver. Instead of this natural scheme, in order to ease the analysis, we describe a slight variation as used for the deterministic channel: receiver $i$ first decodes those interfering signals from user $0$ that appear above the intended signal $x_i$, and only then jointly decodes $x_i$ together with the remaining interfering signals from $X_0(1),\dots,X_0(i)$. Decoding first the interference received above the intended signal $x_i$ corresponds to the fact that in the deterministic channel such interference does not actually interact with the intended signal.

When receiver $i$ is decoding signal $X_0(k)$ for $k\leq i$, assuming the stronger signals $X_0(1),\dots,X_0(k-1)$ have already been decoded and subtracted off, receiver $i$ experiences an effective noise power at most
\begin{equation}\begin{split}\label{e:OTMGauss_EffNoisePower}
N_i(k)=N_0+P_i|h_{ii}|^2+|h_{i0}|^2 \sum_{l>k}\theta_l &= N_0(1+\snr_i)+|h_{i0}|^2 q_{k}
\\ &= N_0\left(1+\snr_i+\frac{\inr_i}{\inr_k}\right)\,.
\end{split}\end{equation}

Similarly to the deterministic case, let $\lambda(i)\in \{1,\dots,i\}$ be such that 
\begin{equation}\label{e:OneToManyG_lambdaj}
  q_{\lambda(i)}|h_{i0}|^2<P_i|h_{ii}|^2\leq q_{\lambda(i)-1}|h_{i0}|^2\,,
\end{equation}or equivalently
\begin{equation}\label{e:OTMG_lambda2}
  \frac{\inr_i}{\inr_{\lambda(i)}}<\snr_i\leq   \frac{\inr_i}{\inr_{\lambda(i)-1}}\,.
\end{equation}
Thus, using \eqref{e:OTMGauss_EffNoisePower} and \eqref{e:OTMG_lambda2}, when receiver $i$ is decoding the signal $X_0(k)$, $k<\lambda(i)$, the effective noise is 
\begin{equation}
  N_i(k) \leq N_0\left(1+2\frac{\inr_i}{\inr_k}\right)\,,
\end{equation}
and hence the effective SNR is at least
\begin{align*}
  \frac{\theta_k|h_{i0}|^2}{N_0\left(1+2\frac{\inr_i}{\inr_k}\right)} = \frac{\theta_k|h_{i0}|^2}{N_0\left(1+2\frac{|h_{i0}|^2}{|h_{k0}|^2}\right)} \geq \frac{\theta_k |h_{k0}|^2}{3N_0}\,,
\end{align*}where the inequality follows from the fact that $|h_{k0}|\leq |h_{i0}|$ for $k\leq i$.
Since this SNR can support the rates of user 0's codebooks given in \eqref{e:OneToManyG_user0rates},
receiver $i$ can decode all the signals $1,\dots,X_0(\lambda(i)-1)$ while treating the signals $x_i$ and $X_0(\lambda(i)),\dots,X_0(K+1)$ as noise. 

It remains to check which rates allow for joint decoding of signals $X_0(\lambda(i)),\dots,X_0(i)$ and $x_i$ by receiver $i$. 
The MAC constraints at receiver $i$ are
\begin{equation}\label{e:OTMGauss_MACconstraints1}
  \sum_{k\in \Lambda}R_0(k)\leq \log\left(1+\frac{\sum_{k\in\Lambda}\theta_k|h_{i0}|^2}{N_0}\right)
\end{equation}
and
\begin{align}
r_i\sum_{k\in \Lambda}R_0(k)\leq \nonumber \log\left(1+\frac{N_0\snr_i+\sum_{k\in\Lambda}\theta_k|h_{i0}|^2}{N_0}\right)\,,\quad
\Lambda\subseteq \{\lambda(i),\dots,i\}\,.\label{e:OTMGauss_MACconstraints2}
\end{align} We may ignore the first set of constraints \eqref{e:OTMGauss_MACconstraints1}: they are readily seen to be satisfied by the choice of rates $R_0(k)$ in \eqref{e:OneToManyG_user0rates}. The second set of constraints \eqref{e:OTMGauss_MACconstraints2} can also be simplified:
it turns out that just as in the deterministic channel, the two constraints for $\Lambda=\{\lambda(i),\dots, i\}$ and $\Lambda=\{\lambda(i)+1,\dots,i\}$ imply the others (up to a small gap). To see this, note that because $\theta_k$ is decreasing in $k$ and by the definition \eqref{e:OneToManyG_lambdaj} of $\lambda(i)$, for $k>\lambda(i)$ it holds that
$$\theta_k|h_{i0}|^2\leq \theta_{\lambda(i)}|h_{i0}|^2 \leq N_0\snr_i\,.$$
Thus for any $\Lambda\subseteq \{\lambda(i),\dots,i\}$ with $\lambda(i)\in \Lambda$,
\begin{align*}
r_i+\sum_{k\in\Lambda} R_0(k)&\leq  
r_i+\sum_{k\in \{\lambda(i),\dots, i\}}R_0(k)\\&\leq 
\log\left(1+\frac{2 q_{\lambda(i)-1}|h_{i0}|^2}{N_0}\right)
\\ &\leq 1+ \log\left(1+\frac{N_0\snr_i+\sum_{k\in\Lambda}\theta_k|h_{i0}|^2}{N_0}\right)\,,
\end{align*}
and similarly for any $\Lambda\subseteq \{\lambda(i),\dots,i\}$ with $\lambda(i)\notin\Lambda$,
\begin{align*}
r_i+\sum_{k\in\Lambda} R_0(k)
&\leq r_i+\sum_{k\in \{\lambda(i)+1,\dots, i\}}R_0(k)\\&\leq
\log\left(1+\frac{2N_0\snr_i}{N_0}\right)
\\ &\leq 1+ \log\left(1+\frac{N_0\snr_i+\sum_{k\in\Lambda}\theta_k|h_{i0}|^2}{N_0}\right)\,.
\end{align*}
Thus, up to a gap of $1$ bit per user (and dropping the one in the logarithms, which only reduces the achievable rate), it is possible to achieve any point in the region determined by the sum-rate constraints
\begin{equation}\label{e:OTMG_achievable1}
  r_i+\sum_{k\in \{\lambda(i),\dots, i\}}R_0(k)\leq 
  \log\left(\frac{q_{\lambda(i)-1}|h_{i0}|^2}{N_0}\right),\quad 1\leq i\leq K\,,
\end{equation}
and
\begin{equation}\label{e:OTMG_achievable2}
  r_i+\sum_{k\in \{\lambda(i)+1,\dots, i\}}R_0(k)\leq 
  \log(\snr_i),\quad 1\leq i\leq K\,,
\end{equation}
together with the individual rate constraints
\begin{equation}\label{e:OTMG_achievable3}
  r_i\leq \log(\snr_i)\,, \quad 1\leq i\leq K\,,
\end{equation}
and
\begin{equation}\label{e:OTMG_achievable4}
R_0(i)\leq \log\left(1+\frac{\theta_i|h_{i0}|^2}{3N_0}\right)\,,\quad 1\leq i\leq K+1\,.
\end{equation}

We now compare the achievable region to the capacity region of the deterministic one-to-many IC. Let $n_{ii}=\snr_i$, $0\leq i\leq K$, and $n_{i0}=\inr_i$, $1\leq i\leq K$. Then the achievable region given by \eqref{e:OTMG_achievable1}, \eqref{e:OTMG_achievable2}, \eqref{e:OTMG_achievable3}, and \eqref{e:OTMG_achievable4} contains the region given by
\begin{equation*}
  r_i+\sum_{k\in \{\lambda(i),\dots, i\}}R_0(k)\leq 
  n_{i0}-n_{\lambda(i)-1,0},\quad 1\leq i\leq K\,,
\end{equation*}
and
\begin{equation*}
  r_i+\sum_{k\in \{\lambda(i)+1,\dots, i\}}R_0(k)\leq 
  n_{ii},\quad 1\leq i\leq K\,,
\end{equation*}
together with the individual rate constraints
\begin{equation*}
  r_i\leq n_{ii}\,, \quad 1\leq i\leq K\,,
\end{equation*}
and
\begin{equation*}
R_0(i)\leq (n_{i0}-n_{i-1,0}-1)^+\,,\quad 1\leq i\leq K+1\,.
\end{equation*}

Comparing this with the deterministic channel region \eqref{e:OTMdet_indconstraint}, \eqref{e:OneToManyDetHK}, and \eqref{e:OneToManyDetHK2}, evidently the regions are the same except that user 0 loses up to 1 bit per signal level, for a total loss of at most $K$ bits. 
Since the outer bound has a gap from the deterministic channel of $K+1$ bits at user $0$, we have determined the capacity region of the one-to-many IC to within a gap of $(2K+1,1,\dots,1)$.
This completes the proof of Theorem~\ref{t:OTMGauss_Capacity}.

\begin{remark}
  Instead of the HK scheme used here, it is possible to use an achievable scheme that creates independent levels, and then to emulate the first scheme presented for the deterministic one-to-many channel. However, such an approach yields a larger gap between the inner and outer bounds.
\end{remark}

As with the many-to-one channel, the generalized degrees of freedom can now be computed.
\begin{theorem}
    Put $\snr_i=s^{\al_i}$ for $0\leq i\leq K$ and $\inr_i=s^{\beta_i}$ for $1\leq i \leq K$.
  The degrees-of-freedom region of the one-to-many Gaussian IC is the set of points satisfying the individual constraints $$d_i\leq \al_i,\quad 0\leq i\leq K,$$ and the sum-rate constraints (for each set of users relabel the users as $\{1,\dots,m\}$)
  \begin{align*}\sum_{i=0}^m d_i  \leq (\al_0-\beta_m)^+ + \max(\al_1,\beta_1)+\sum_{i=2}^m \max(\al_i,\beta_i-\beta_{i-1})\, .\end{align*}
\end{theorem}

\section{Conclusion} \label{s:conclusion}
In finding the capacity of the many-to-one and one-to-many Gaussian interference channels, two main themes emerge: the power of the deterministic model approach, and the use of lattice codes for interference alignment. Throughout the entire development, the deterministic model serves as a faithful guide to the Gaussian channels. The structure of the outer bound, namely the existence of sum-rate constraints for every subset of users, is most easily observed in the deterministic channel. Moreover, the proofs of the Gaussian outer bounds closely follow those for the deterministic channels, with the side information used to prove outer bounds in the Gaussian case translated directly from the deterministic case. 

The capacity achieving schemes are very simple in the deterministic channels. 
The interference alignment phenomenon emerges in the deterministic many-to-one channel, but in order to translate the scheme to the Gaussian channel, lattices are necessary in order to provide an alignment in signal scale. 
Yet another success of the deterministic model is that the reciprocity between the many-to-one and one-to-many channels is evident in the deterministic setting; this basic relationship between the two channels is veiled in the Gaussian case. 

The approach used here should be contrasted with the direct approach of \cite{BT08}, where the problem of finding the capacity of the 2-user Gaussian IC to within a constant gap was \emph{reduced} to that of finding the capacity of a corresponding deterministic channel. More generally, the limitations and potential of the deterministic approach beg to be studied.

The gap of $(2K+5)\log K$ bits per user between the achievable region and outer bound in the many-to-one Gaussian IC (Theorem~\ref{t:smallGap}) is somewhat loose. One way that the bound could be improved is to account for the combinatorial structure of the interference pattern (Figure~\ref{GaussianRectangles}) in evaluating the achievable strategy. There is a balance between the number of intervals formed by the interfering signals and the loss required by each user due to addition of signals from lower levels. The outer bound can probably also be tightened by more carefully performing the estimate in \eqref{e:manyToOneErrorTerm}.



\appendices


\section{Gaussian Han and Kobayashi achieves sum-rate of at most $\log(1+3\beta^2)$.}\label{s:GaussianHKcalculation}
This section contains a proof of Claim~\ref{l:HKsuboptimal}, showing that a Han-Kobayashi scheme with Gaussian codebooks cannot achieve a sum-rate greater than $\log(1+3\beta^2)$.

At an achievable rate point, receiver 0 is assumed to be able to decode message 0. After decoding, receiver 0 may subtract away signal 0; since users 1 and 2 use a superposition codebook with private and common messages in the Han-Kobayashi scheme \cite{HanKobayashi}, there are four messages which receiver 0 should be able to decode. Let the four (Gaussian) codebooks have rates $r_1^a,r_1^b,r_2^a,r_2^b$ ($r_1=r_1^a+r_1^b$ and $r_2=r_2^a+r_2^b$) and received power to noise ratios $\beta S_1,\beta(1-S_1),\beta S_2,\beta(1-S_2)$, respectively, at the intended receivers. Receivers 1 and 2 are assumed to be able to decode their own signals, so the MAC constraints at receiver $i=1,2$ hold:\begin{equation}\begin{split}\label{e:individualMAC}
  r_i^a & \leq \log(1+ \beta S_i)
  \\
  r_i^b & \leq \log(1+ \beta (1-S_i))
  \\
  r_i^a+r_i^b &\leq \log (1+\beta)\, .
\end{split}
\end{equation}
Now, consider the MAC constraints at receiver 0. It is assumed that $\beta\geq 2$. The received power to noise ratios are each scaled by $\beta$, since the gains are $h_{01}=h_{02}=\beta$ as compared to the gains $h_{11}=h_{22}=\sqrt{\beta}$ on links 1 and 2. Thus, the constraints on $(r_i^a,r_i^b)$, for $i=1,2$ separately, are obviously satisfied at receiver 0. 

To check that the constraint on $r_1^a+r_2^a$ is satisfied, note that by the constraints at receivers 1 and 2 \eqref{e:individualMAC}, \begin{equation}\begin{split}\label{e:privateSumRate}
  r_1^a+r_2^a&\leq \log(1+\beta S_1)+\log(1+\beta S_2)
  \\ &= \log(1+\beta S_1+\beta S_2+\beta^2 S_1 S_2)
  \\ &\leq \log \left(1+\beta^2 \frac{S_1+S_2}{2}\right)\, ,\end{split}
\end{equation} where the last step follows from the inequality $(S_1+S_2)/2\geq S_1 S_2$ for $0\leq S_1,S_2\leq 1$ and $\beta^2 /2 \geq \beta$ for $\beta\geq 2$. 
Defining $S_i'= 1-S_i$, equation \eqref{e:privateSumRate} shows that the constraint on $r_1^b+r_2^b$ is also satisfied at receiver 0. 
Similarly, the constraint on $r_1^a+r_2^b$ (as well as $r_1^b+r_2^a$) is satisfied:
\begin{align*}
  r_1^a+r_2^b&\leq \log(1+\beta S_1)+\log(1+\beta(1-S_2))
  \\ &= \log (1+\beta S_1+\beta(1-S_2)+\beta^2 S_1 (1-S_2))
  \\&\leq \log (1+\beta^2 S_1+\beta^2 (1-S_2))\,.
\end{align*}
Continuing, for the constraint on $r_1^a+r_1^b+r_2^a$, we have
\begin{equation*}
  \begin{split}
    r_1^a+r_1^b+r_2^a &\leq \log(1+\beta)+\log(1+ \beta S_2)
    \\&= \log(1+\beta^2 S_2+\beta+\beta S_2)
    \\&\leq \log(1+\beta^2 S_2+ \beta^2) \, .
  \end{split}
\end{equation*} Again by symmetry, all constraints on 3 rates are seen to be satisfied at receiver 0. The last remaining sum-rate constraint is also satisfied:
\begin{equation*}
  \begin{split}
    r_1^a+r_1^b+r_2^a+r_2^b &\leq 2 \log (1+\beta)
    \\ &= \log (1+2\beta+ \beta^2)
    \\ &\leq \log (1+2\beta^2)\, .
  \end{split}
\end{equation*}
Thus, because receiver 0 can decode all three messages, the MAC constraints apply, and the sum-rate achieved by a Gaussian Han and Kobayashi scheme is upper bounded as $$r_\text{sum}^{HK}\leq \log(1+3\beta^2).$$

\section{Proof of Lemma~\ref{l:achieveCompatible}}
We prove the contrapositive of the statement of Lemma~\ref{l:achieveCompatible}:

\noindent\emph{Let $\vec r$ be some rate allocation achieving with equality two incompatible constraints on sets $A,A^\pr$ from Theorem~\ref{t:DeterministicCapacity}. Then there exists another constraint from Theorem~\ref{t:DeterministicCapacity} that is violated by $\vec r$.}

Suppose the constraints on two sets (vertices) $A$ and $A^\pr$ on the left-hand side of the bipartite graph are incompatible, i.e. in the bipartite graph some right-hand vertex $k^*$ is a solid neighbor of $A^\pr$ and a dashed neighbor of $A$ (or vice-versa). Furthermore, suppose the rate point $\vec r=\{r_0,\ldots,r_K\}$ achieves the constraints on $A$ and $A^\pr$ with equality. We shall assume that this rate point satisfies all constraints from Theorem~\ref{t:DeterministicCapacity} and derive a contradiction.

First, recall that vertices corresponding to individual rate constraints have only solid edges. Thus, $A$ must correspond to a sum-rate constraint, as it is assumed to have at least one dashed edge.
If the vertex $A^\pr$ corresponds to an individual rate constraint, it is straightforward to show that the given rate point $\vec r$ violates the constraint on $A\cup A^\pr$. Thus, we consider sets of users $A=\S\cup \{0\}$ and $A^\pr=\S^\pr\cup \{0\}$, where $\S,\S^\pr\subseteq \{1,\ldots, K\}$ and both $\S$ and $\S^\pr$ are nonempty.

A user $i$ is said to be \emph{occluded} by a set of users $I$ if for each $k$ with $i\in U_k$, $|U_k\cap (I\setminus \{i\})|\geq 1$. This means each level (at receiver 0) at which user $i$ can interfere is occupied by at least one user in $I\setminus \{i\}$. If the sum-rate constraint on $I$ is met with equality, and $i\in I$ is occluded by $I$, it is straightforward to show that the sum rate constraint on $I\setminus \{i\}$ is also met with equality.
If there is an occluded user $i\in \S^\pr$ interfering on level $\badk$, then, as in the case where $\S^\pr$ corresponds to an individual rate constraint, the constraint on the set $\{i\}\cup \S $ is violated. If there are three or more users sharing any set of levels, it is easy to show that one of them must be occluded by the other two. Hence, we can assume there are exactly two users $a$ and $b$ in $\S^\pr$ interfering on level $\badk$, i.e. $U_\badk \cap \S^\prime =\{a,b\}$. Choose $a,b$ so that $n_{0a}-n_{aa}<n_{0b}-n_{bb}$, and assume $n_{0b}>n_{0a}$ (otherwise $a$ occludes $b$). Moreover, by the previous statements, we can assume that no level contains more than two users from each of $\S$ and $\S^\pr$.

Let $\S_a=\{i\in \S: n_{0i}<\badk\}$ be the set of users in $\S$  whose interference at receiver 0 occurs at levels below $\badk$, and let $\S_b=\S\setminus \S_a$ be the remaining users in $\S$. Note that $\S_a$ and $\S_b$ occupy disjoint sets of levels, because level $\badk$ separates $\S_a$ and $\S_b$.
Similarly, let $\S_a^\prime=\{a\}\cup\{i\in \S^\prime: n_{0i}\leq \badk\}$ be the set of users in $\S^\prime$ whose interference at receiver zero occurs at or below user $a$, and let $\S_b^\prime=\S^\prime\setminus \S_a^\prime$. Figure~\ref{UserSelectionFig} depicts the relationship between the various sets.

\begin{figure}
\psset{unit=.6mm,arrowlength=1.2,arrowinset=0,labelsep=3pt}
\begin{center}
\begin{pspicture}(-1,-4)(113,82)
\psline[linestyle=dashed,dash=3pt 2pt](0,0)(112,0)
\psline(4,0)(4,75)(12,75)(12,0) \rput(8,-5){0}
\pspolygon(28,16)(28,36)(36,36)(36,16)\rput(32,26){\small $a$}
\pspolygon(52,30)(52,52)(60,52)(60,30)\rput(56,41){\small $b$}
\pspolygon(76,19)(76,31)(84,31)(84,19)
\pspolygon(96,35)(96,48)(104,48)(104,35)

\psellipse[linestyle=dotted](29,19)(14,22) \rput(24,10){\small
$\S^\prime_a$} \psellipse[linestyle=dotted](54,52)(14,25)
\rput(55,62){\small $\S^\prime_b$}
\psellipse[linestyle=dotted](80,18)(14,18)  \rput(80,10){\small
$\S_a$} \psellipse[linestyle=dotted](100,50)(14,21)
\rput(100,60){\small $\S_b$}

\pspolygon*(28,32)(28,34)(36,34)(36,32)
\pspolygon*(52,32)(52,34)(60,34)(60,32)
\pspolygon*(4,32)(4,34)(12,34)(12,32)   \rput(0,33){$\badk$}

\end{pspicture}
\end{center}
\caption{This figure illustrates the choice of sets $S_a,S_b,S_a^\pr,S_b^\pr$. Notice that both users $a$ and $b$ in $\S^\pr=\S_a^\pr\cup \S_b^\pr$ interfere on level $\badk$, while no user in $\S=\S_a\cup \S_b$ interferes on level $\badk$.}
\label{UserSelectionFig}
\end{figure}
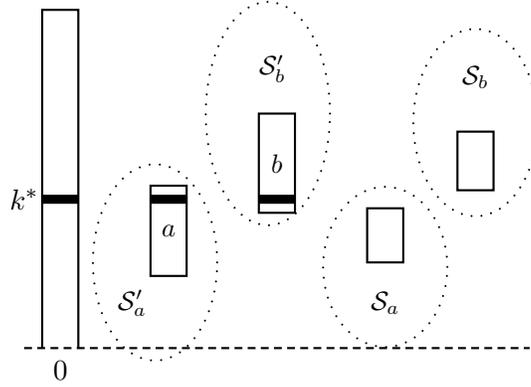

Denote by $f(\S)$ the value of the sum-rate constraint on $\S$, i.e. $$f(\S)=f_\text{free}(\S)+\sum_{k=1}^{\al_{00}}f_k(\S)\,.$$
It is given in the statement of the lemma that the sum-rate constraints on $\S\cup\{0\}$ and $\S^\pr \cup \{0\}$ are met with equality: \begin{equation}\label{ConstraintOnSetS}r_0+\sum_{i\in \S_a}r_i +\sum_{i\in \S_b} r_i  = f(\S)\end{equation}and \begin{equation}\label{ConstraintOnSetSPrime}r_0+\sum_{i\in \S_a^\prime}r_i +\sum_{i\in \S_b^\prime} r_i  = f(\S^\prime).\end{equation}
Next, we may assume that the constraints on $\S_b\cup \S_a^\pr\cup\{0\}$ and $\S_a\cup \S_b^\pr\cup \{0\}$ are satisfied, i.e. \begin{subequations}\label{ViolatedConstraintAssumptions}\begin{equation}\sum_{i\in \S_b} r_i +\sum_{i\in \S_a^\prime}r_i+r_0\leq f(\S_b\cup \S_a^\prime)\end{equation} and \begin{equation} \sum_{i\in \S_a} r_i +\sum_{i\in \S_b^\prime}r_i+r_0\leq f(\S_a\cup \S_b^\prime);\end{equation}\end{subequations} otherwise, we have the desired violated constraint.
Plugging these two inequalities into equation (\ref{ConstraintOnSetS}) gives \begin{equation*}r_0 + f(\S_b\cup \S_a^\prime) + f(\S_a\cup \S_b^\prime)-r_0-\sum_{i\in \S_a^\prime}r_i-r_0-\sum_{i\in \S_b^\prime} r_i\geq f(\S),\end{equation*}or, upon rearranging,\begin{equation*}f(\S_b\cup \S_a^\prime) + f(\S_a\cup \S_b^\prime)-f(\S)\geq r_0+\sum_{i\in \S_a^\prime}r_i+\sum_{i\in \S_b^\prime} r_i.\end{equation*} This, with equation (\ref{ConstraintOnSetSPrime}), implies that \begin{equation} \label{ContradictionEquation} f(\S_b\cup \S_a^\prime) + f(\S_a\cup \S_b^\prime)-f(\S) \geq f(\S^\prime),\end{equation} which, as we show next, is a contradiction.
The definition of $f(\cdot)$, the fact that $\S_a\cap \S_b= \S_a\cap \S_b^\pr=\S_b\cap \S_a^\pr=\varnothing$, and the fact that $\S_a$ and $\S_b$ occupy disjoint sets of levels gives
\begin{align}&\left(f(\S_b\cup \S_a^\prime) + f(\S_a\cup \S_b^\prime)-f(\S)\right)       \nonumber
\\& =
   \sum_{i\in \S_b\cup \S_a^\pr}\left( (n_{i0}-n_{ii})^+ +(n_{i0}-n_{00})^+\right)+
   n_{00}
+\sum_{k=1}^\an (|U_k\cap (\S_b\cup \S_a^\prime)|-1)^+     \label{2ndTerm}
 \\&\quad+
   \sum_{i\in \S_a\cup \S_b^\pr}\left( (n_{i0}-n_{ii})^+ +(n_{i0}-n_{00})^+\right)+
   n_{00}+\sum_{k=1}^\an (|U_k\cap (\S_a\cup \S_b^\prime)|-1)^+       \label{4thTerm}
\\ &\quad -
   \sum_{i\in (\S_a\cup \S_b)} \left( (n_{i0}-n_{ii})^+ +(n_{i0}-n_{00})^+\right) - n_{00}\nonumber
-\sum_{k=1}^\an (|U_k\cap (\S_a\cup \S_b)|-1)^+
   \nonumber
\\&=
   \sum_{i\in \S_a^\pr\cup \S_b^\pr}\left( (n_{i0}-n_{ii})^+ +(n_{i0}-n_{00})^+\right)+n_{00}\nonumber
\\&\quad
   + \sum_{k=1}^\an  (|U_k \cap (\S_b\cup \S_a^\pr)|-1)^+  \nonumber -\sum_{k=1}^\an (|U_k\cap \S_a|-1)^+
\\ &\quad +\sum_{k=1}^\an  (|U_k \cap (\S_a\cup \S_b^\pr)|-1)^+
   -\sum_{k=1}^\an (|U_k\cap \S_b|-1)^+ \label{manysums}
\end{align}
To continue, we rewrite the second term in (\ref{2ndTerm}) as
\begin{align}\nonumber
&\sum_{k=1}^\an  (|U_k \cap (\S_b\cup \S_a^\pr)|-1)^+
\\ & =
   \sum_{k:|U_k\cap \S_b|\neq 0}  (|U_k \cap \S_b|-1) +\sum_{k:|U_k\cap \S_b|\neq 0}|U_k\cap \S_a^\pr|\nonumber +
   \sum_{k: |U_k\cap \S_b|= 0}(|U_k\cap(\S_b\cup \S_a^\pr)|-1)^+ \nonumber
\\ & =
   \sum_{k=1}^\an (|U_k\cap \S_b|-1)^+
   +\sum_{k:|U_k\cap \S_b|\neq 0}|U_k\cap \S_a^\pr| +
  \sum_{k: b\notin U_k}(|U_k\cap \S_a^\pr|-1)^+ ,
\label{expression1}
\end{align}
where the last step follows from 1) the condition $|U_k\cap \S_b|= 0$ underneath the third summation; and 2) the observation that for $k$ such that $|U_k\cap \S_a^\pr|\neq 0,$ $|U_k\cap \S_b|\neq 0$ implies $b\in U_k$, and for $b\in U_k$, it holds that $|U_k\cap \S_a^\pr|\leq 1$, so $(|U_k\cap \S_a^\pr|-1)^+=0$.
Similarly, for the second term in (\ref{4thTerm}) we have \begin{align}\nonumber
&\sum_{k=1}^\an  (|U_k \cap (\S_a\cup \S_b^\pr)|-1)^+
\\ & =
   \sum_{k:|U_k\cap \S_a|\neq 0}  (|U_k \cap \S_a|-1) +\sum_{k:|U_k\cap \S_a|\neq 0}|U_k\cap \S_b^\pr|+
   \sum_{k: |U_k\cap \S_a|= 0}(|U_k\cap(\S_a\cup \S_b^\pr)|-1)^+ \nonumber
\\& =
   \sum_{k=1}^\an (|U_k\cap \S_a|-1)^+
   +\sum_{k:|U_k\cap \S_a|\neq 0}|U_k\cap \S_b^\pr|+
  \sum_{k: a\notin U_k}(|U_k\cap \S_b^\pr|-1)^+ .
\label{expression2}
\end{align}
Plugging (\ref{expression1}) and (\ref{expression2}) into equation (\ref{manysums}) and canceling terms results in the expression
\begin{align}
   &\sum_{i\in \S_a^\pr\cup \S_b^\pr}\left( (n_{i0}-n_{ii})^+ +(\alpha
   _{i0}-n_{00})^+\right)+n_{00}\nonumber
\\&\quad+\sum_{k:|U_k\cap \S_a|\neq 0}|U_k\cap \S_b^\pr|  \nonumber +
   \sum_{k: a\notin U_k}(|U_k\cap \S_b^\pr|-1)^+
    \\&\quad +\sum_{k:|U_k\cap \S_b|\neq 0}|U_k\cap \S_a^\pr|+
   \sum_{k: b\notin U_k}(|U_k\cap \S_a^\pr|-1)^+. \label{temp1}
\end{align}
By performing manipulations in the same style as above, it is possible to write \begin{align*} f(\S^\pr)
&=
   \sum_{i\in \S_a^\pr\cup \S_b^\pr}\left( (n_{i0}-n_{ii})^+     +   (\alpha
   _{i0}-n_{00})^+\right)+n_{00}
\\&\quad+   |\{k:\{a,b\}\subseteq U_k\}|
   +    \sum_{k: b\notin U_k}(|U_k\cap \S_a^\pr|-1)^+
\\&\quad +
   \sum_{k: a\notin U_k}(|U_k\cap \S_b^\pr|-1)^+.\end{align*}
Comparing this with the previous expression (\ref{temp1}), we conclude that
\begin{align*}&f(\S_b\cup \S_a^\prime) + f(\S_a\cup
   \S_b^\prime)-f(\S)
\\ & =
   f(\S^\prime)+
   -|\{k:\{a,b\}\subseteq U_k\}| +
   \sum_{k:|U_k\cap \S_a|\neq 0}|U_k\cap \S_b^\pr|
   +\sum_{k:|U_k\cap \S_b|\neq 0}|U_k\cap \S_a^\pr|
\\ & =
   f(\S^\prime) -|\{k:\{a,b\}\subseteq U_k\}|
   +    |\{k:b\in U_k,|U_k\cap \S_a|\neq 0\}|
+
   |\{k:a\in U_k,|U_k\cap \S_b|\neq 0\}|
\\ &
   \leq f(\S^\pr)-1.
\end{align*}
But this contradicts equation (\ref{ContradictionEquation}), proving the lemma.


\section{Proof of the Sum-Rate Constraint for the Many-to-One Gaussian Channel}
The proof of the sum-rate constraint uses a genie-aided channel, or in other words, allows the receivers access to side information. The main difficulty of the proof lies in choosing this side information. The crucial insight is provided by the deterministic channel model. Recall the side information given to receiver 0 in the many-to-one deterministic IC \eqref{e:manyToOneSideInf}; there, on each level receiver 0 was given the signals of all interfering users except for one. 
From Figure~\ref{fig:GaussianManyToOneSideInf}, we see that this side information corresponds exactly to giving the top portion of each interfering signal. Informed by the analogy that additive Gaussian noise corresponds to truncation in the deterministic channel (see Figure~\ref{fig:p2pDeterministic}),
we give side information
\begin{align*}s_0&=\sig_{m}
\\s_k &= (\sum_{i=1}^k h_{io} x_i+z_0,\sig_k),\quad 1\leq k\leq m\, ,\end{align*}
where for each $k$, $1\leq k\leq m$, we have
$$\sig_k=(h_{10}x_1+w_1+z_0,h_{20}x_2+w_2,\dots,h_{k0}x_{k}+w_{k})
$$ with $\sig_0=0$
and
\begin{align*}
  &w_i\sim\CN(0, \n \max( \inr_{i+1}/\snr_{i+1},1)),\quad 1\leq i\leq m-1 \\ &w_m\sim\CN(0,\n \snr_0).
\end{align*}

\begin{figure}
\begin{centering}
\psset{unit=.8mm,arrowlength=1,arrowinset=0,labelsep=2pt}
\begin{center}
\begin{pspicture}(-1,-2)(87,62)
\psline[linestyle=dashed,dash=3pt 2pt](0,0)(86,0)
\psline(4,0)(4,60)(12,60)(12,0) \rput(8,-5){0}
\pspolygon(28,0)(28,34)(36,34)(36,0)\rput(32,-5){1}
\pspolygon(52,20)(52,48)(60,48)(60,20)\rput(56,-5){2}
\pspolygon(76,42)(76,57)(84,57)(84,42)\rput(80,-5){3}

\pspolygon[fillstyle=solid, fillcolor=gray](28,20)(28,34)(36,34)(36,20)

\pspolygon[fillstyle=solid, fillcolor=gray](52,42)(52,48)(60,48)(60,42)

\uput[r](12,60){\footnotesize $\log \snr_0$} \uput[u](40,34){\footnotesize
    $\log \inr_1$}
\uput[u](65,48){\footnotesize $\log \inr_2$} \uput[r](60,20){\footnotesize{
$\log\frac{\inr_2}{\snr_2}$}} \uput[u](88,57){\footnotesize
$\log\inr_3$} \uput[d](88,42){\footnotesize
$\log\frac{\inr_3}{\snr_3}$}

\end{pspicture}
\end{center}
\caption{The side information given to user $0$ is shaded. The side information is precisely the portion of each signal overlapping with the next signal.}\label{fig:GaussianManyToOneSideInf}
\end{centering}
\end{figure}
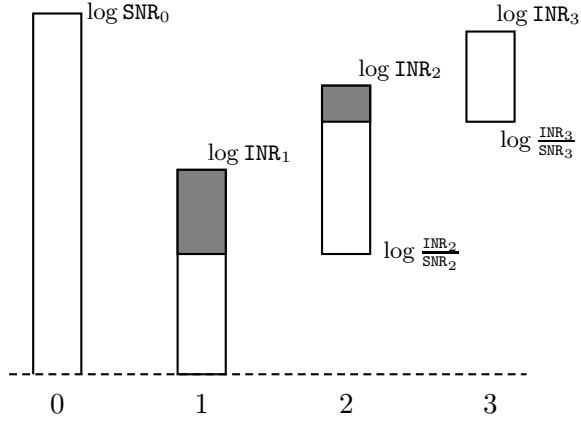

With this choice of side information the proof is straightforward albeit fairly lengthy.
Fano's inequality and the data processing inequality
imply that
\begin{align}\label{GaussianSumRateFano}
N(r_0 + r_1+\dots+r_m-\eps_N)
&\leq
    \sum_{i=0}^m I(x_i^N; y_i^N,s_i^N),
\end{align}
where $\eps_N\to 0$ as $N\to \infty$.
Each term in the sum can be expanded as $$I(x_i^N; y_i^N,s_i^N)= h(y_i^N|s_i^N)+h(s_i^N)-h(y_i^N,s_i^N|x_i^N).$$
Using the fact that $x_0$ is independent of $\sig_{m}$ and $x_k$ is independent of $\sig_{k-1}$ for $1\leq k\leq m$, the negative terms evaluate to
\begin{align*}h(y_0^N,s_0^N|x_0^N)
&=  h\left(h_{00}x_0^N+\sum_{i=1}^m
    h_{i0}x_i^N+z_0^N,\sig_{m}^N|x_0\right)
\\&\quad=  h(\sum_{i=1}^m
    h_{i0}x_i^N+z_0^N,\sig_{m}^N)
=    h(s_m^N),
\\
h(y_1^N,s_1^N|x_1^N)
&= h(h_{11}x_1^N+z_1^N,h_{01}x_1^N+z_0^N,h_{10}x_1+w_1^N+z_0^N
|x_1^N)
\\&\quad =  h(z_1^N,z_0^N,w_1^N)
\\&\quad =  h(z_1^N)+h(z_0^N)+h(w_1^N)
\end{align*}
and for
$2\leq k\leq m$,
\begin{align*}h(y_k^N,s_k^N|x_k^N)
&=
    h(h_{kk}x_k^N+z_k^N,\sum_{i=1}^k h_{i0}
    x_i^N+z_0^N,\sig_k^N|x_k^N)
\\&=
    h(z_k^N)+h(\sum_{i=1}^{k-1}
    h_{i0}x_i^N+z_0^N,\sig^N_{k-1},w_k^N)
\\&=
    h(z_k^N)+h(s_{k-1}^N)+h(w_k^N).
\end{align*}
The sum in equation (\ref{GaussianSumRateFano}) telescopes, giving
\begin{align}
N(r_0 + r_1+\cdots+r_m-\eps_N) \nonumber
&\leq
    \sum_{i=0}^m h(y_i^N|s_i^N)+h(s_0^N)-h(y_1^N,s_1^N|x_1^N) \nonumber
\\&\quad +
    h(s_m^N)-h(y_0^N,s_0^N|x_0^N)
+\sum_{i=1}^{m-1} \left[h(s_i^N)-h(y_{i+1}^N,s_{i+1}^N|x_{i+1}^N)\right] \nonumber
\\&=  \label{GaussianSumRateTelescope}
    \sum_{i=0}^m \bigg[h(y_i^N|s_i^N)-h(z_i^N)\bigg]-\sum_{i=1}^{m}h(w_i^N)+h(s_0^N).
\end{align}
Next, we bound each term using the independence bound on entropy, and the fact that the Gaussian distribution maximizes differential entropy for a fixed (conditional) variance.
%

\begin{fact}[Worst-case Conditional Entropy]
Let $z_1\sim \CN(0,\sig_1^2)$, $z_2\sim \CN(0,\sig_2^2)$, and $x$ be mutually independent with $E(|x|^2)\leq P$. Then \begin{equation} h(x+z_1|x+z_2)\leq \log \left[\pi e\left(\sig_1^2+\frac{P\sig_2^2}{P+\sig_2^2}\right)\right].\end{equation} \begin{proof}\begin{align*}h(x+z_1|x+z_2)
&=h(x+z_1-\alpha(x+z_2)|x+z_2)\\
&\leq h(z_1+x(1-\alpha)-\alpha z_2)
\\&\leq \log \left[\pi e\left(\sig_1^2+P(1-\alpha)^2+\alpha^2 \sig_2^2\right)\right] 
\\&= \log \left[\pi e\left(\sig_1^2+P\frac{\sig_2^2}{(P+\sig_2^2)^2}+\frac{P^2}{(P+\sig_2^2)^2} \sig_2^2\right)\right]
\\&= \log \left[\pi e\left(\sig_1^2+\frac{P\sig_2^2}{P+\sig_2^2}\right)\right]
\end{align*}
where the second to last equality follows by choosing $\alpha=P/(P+\sig_2^2)$.
\end{proof}
\end{fact}

%

We have
\begin{align}
h(s_0^N)&\leq
    \sum_{j=1}^N\left[\sum_{k=2}^{m}h(h_{k0}x_{k,j}+w_{k,j})+h(h_{10}x_{1,j}+z_{0j}+w_{1,j})\right]\nonumber
\\&\leq
    \sum_{j=1}^N\bigg[\sum_{k=2}^{m}\log\left[ \pi e \left( |h_{k0}|^2
    P_{k,j}+P_{w_k} \right)\right]
+\log\left[ \pi e \left( |h_{10}|^2 P_{1,j}+\n+P_{w_1}\right)\right]\bigg]\nonumber
\\&\leq
    N\sum_{k=2}^{m}\bigg\{\log\left[ \pi e \left( |h_{k0}|^2
    \frac{1}{N}\sum_{j=1}^N P_{k,j}+P_{w_k} \right)\right]
\\&\quad +\log\left[\pi e\left(|h_{10}|^2 \frac{1}{N}\sum_{j=1}^NP_{1,j}+\n+P_{w_1}\right)\right]\bigg\},\nonumber
\end{align}
where $P_{k,j}=E|x_{k,j}|^2$.
Jensen's inequality, the power constraint $\frac{1}{N}\sum_j P_{k,j}\leq P_k$, and the fact that $\log x$ is an increasing function justify the remaining steps, continuing from above.
\begin{align}
&\leq
N\sum_{k=2}^{m}\log\left[ \pi e \left( |h_{k0}|^2
    P_{k}+P_{w_k} \right)\right]
+n\log\left[\pi e\left(|h_{10}|^2 P_1+\n +P_{w_1}\right)\right] \nonumber
\\&= \nonumber
    N\Bigg[\sum_{k=2}^{m-1}
    \log\left(
    \inr_k+\max(\frac{\inr_{k+1}}{\snr_{k+1}},1)\right)
+\log\left(1+\inr_1 +\max(\frac{\inr_2}{\snr_2},1)\right)\Bigg]
\\&\quad+N\left[\log(\inr_m+\snr_0) +m\log(\pi
    e \n)\right]\nonumber
\\&\leq \label{GaussianTerm1}
    N\Bigg[1+\sum_{k=1}^{m-1}\log\left(
    \inr_k+\max(\frac{\inr_{k+1}}{\snr_{k+1}},1)\right)
+\log(\inr_m+\snr_0) +m\log(\pi
    e \n)\Bigg],
\end{align}
and similarly
for $2\leq k\leq m$,
\begin{align}h(y_k^N|s_k^N)
&\leq
    \sum_{j=1}^N h(y_{k,j}|s_{k,j}) \nonumber
\\&=
    \sum_{j=1}^N h(h_{kk}x_{k,j}+z_{k,j}|\sum_{i=1}^k h_{i0}x_{i,j}+z_0,\sig_{k,j})\nonumber
\\&\leq
    \sum_{j=1}^N h(h_{kk} x_{k,j}+z_{k,j} | h_{k0} x_{k,j} -\sum_{i=1}^{k-1}w_{i,j})\nonumber
\\&=
    \sum_{j=1}^N h(h_{kk} x_{k,j}+z_{k,j} | h_{kk} x_{k,j} -\sum_{i=1}^{k-1} \frac{h_{kk}}{h_{k0}}w_{i,j}) \nonumber
\\ &\leq
    \sum_{j=1}^N \log\Bigg[ \pi e \bigg(\n 
  +\frac{P_{k,j} |h_{kk}|^2(\n \frac{|h_{kk}|^2}{|h_{k0}|^2}) \sum_{i=2}^{k} \max(\frac{\inr_{i}}{\snr_{i}},1)}{P_{k,j}|h_{kk}|^2+(\n \frac{|h_{kk}|^2}{|h_{k0}|^2}) \sum_{i=2}^{k} \max(\frac{\inr_{i}}{\snr_{i}},1)} \bigg)\Bigg] \nonumber
\\&\leq
     N \log\Bigg[ \pi e \bigg(\n
    + \frac{P_{k} |h_{kk}|^2(\n |h_{kk}|^2/|h_{k0}|^2) \sum_{i=2}^{k} \max(\frac{\inr_{i}}{\snr_{i}},1)}{P_{k}|h_{kk}|^2+(\n |h_{kk}|^2/|h_{k0}|^2) \sum_{i=2}^{k} \max(\frac{\inr_{i}}{\snr_{i}},1)} \bigg)\Bigg] \nonumber
\\&= \label{GaussianTerm2}
    N\log\left(1+\frac{\snr_k\left(\frac{\snr_k}{\inr_k}\sum_{i=2}^k \max(\frac{\inr_{i}}{\snr_{i}},1)\right)}{\snr_k+\left(\frac{\snr_k}{\inr_k}\sum_{i=2}^k \max(\frac{\inr_{i}}{\snr_{i}},1)\right)}\right) +
    N\log(\pi e \n).
\end{align}
Likewise,
\begin{align}
h(y_1^N|s_1^N)\nonumber
\leq
    N\log\left(1+\frac{ \snr_1 \cdot \left(\frac{\snr_1}{\inr_1}+1\right)}{\snr_1+\frac{\snr_1}{\inr_1}+1}\right)+
    N\log(\pi e \n)
\end{align}
and
\begin{align}\label{GaussianTerm3}
h(y_0^N|s_0^N)
\leq  N\log\left[2\snr_0+\sum_{i=1}^{m-1}
\max\left(\frac{\inr_{i+1}}{\snr_{i+1}},1\right)\right]+N\log(\pi e \n)
\end{align}
Finally, by the definition of $w_i$,
\begin{equation}\label{GaussianTerm4}
\begin{split}
h(w_i)&=    \log\left[\pi e \n
    \max(\frac{\inr_{i+1}}{\snr_{i+1}},1)\right],\quad 1\leq i\leq m-1. \\
h(w_m)&=    \log\left(\pi e \n
\snr_0\right).
\end{split}
\end{equation}
Plugging equations (\ref{GaussianTerm1}-\ref{GaussianTerm4}) into (\ref{GaussianSumRateTelescope}), we have the desired sum-rate bound:
\begin{equation}
\begin{split}\label{Gaussian_sum_rate}
&r_0+ r_1 +\dots +r_m
\\& \leq   1+ \sum_{k=1}^{m-1}\log  \left[
    \inr_k+\max\left(\frac{\inr_{k+1}}{\snr_{k+1}},1\right)\right]
+\log\left(\inr_m+\snr_0\right)
\\&\quad+
    \sum_{k=1}^m \log\left[1+\frac{\snr_k    \left(  \frac{\snr_k}{\inr_k}\sum_{i=2}^k \max\left(\frac{\inr_{i}}{\snr_{i}},1\right)\right)}{\snr_k+
    \left(\frac{\snr_k}{\inr_k}\sum_{i=2}^k \max\left(\frac{\inr_{i}}{\snr_{i}},1\right)   \right)}\right]
\\&\quad+
\log\left[2\snr_0+\sum_{i=1}^{m-1}
\max\left(\frac{\inr_{i+1}}{\snr_{i+1}},1\right)\right]
\\&\quad-
\sum_{k=1}^{m-1} \log\left[
\max\left(\frac{\inr_{k+1}}{\snr_{k+1}},1\right)\right]-\log(\pi e \n \snr_0).
\end{split}\end{equation}

The structure of the outer bound is not clear from equation (\ref{Gaussian_sum_rate}); therefore, we loosen the constraints in order that their form resemble the deterministic channel constraints.  Beginning with the second sum, consider two cases: $\frac{\inr_k}{\snr_k}\geq 1$, or $\frac{\inr_k}{\snr_k}< 1$. In the first case we proceed as follows. By the assumed ordering on the users, $\inr_i/\snr_i\leq \inr_{i+1}/\snr_{i+1}$ for $1\leq i\leq m-1$, and hence $$\frac{\snr_k}{\inr_k}\sum_{i=2}^k \max\left(\frac{\inr_{i}}{\snr_{i}},1\right)\leq k-1.$$ In the second case,
$$\frac{\snr_k}{\inr_k}\sum_{i=2}^k \max\left(\frac{\inr_{i}}{\snr_{i}},1\right)= \frac{\snr_k}{\inr_k} (k-1).$$
The second sum can therefore be bounded as
\begin{align}
\sum_{k=1}^m \log\left(1+\frac{\snr_k    \left(  \frac{\snr_k}{\inr_k}\sum_{i=2}^k \frac{\inr_{i}}{\snr_{i}}\right)}{\snr_k+
    \left(\frac{\snr_k}{\inr_k}\sum_{i=2}^k \frac{\inr_{i}}{\snr_{i}}   \right)}\right) \nonumber
&\leq       \sum_{k=1}^m \left(\log
    k+\log\bigg(\frac{\snr_k}{\inr_k}\bigg)^+ \right)  \nonumber
\\&\leq     \int_1^{m+1} \log x dx +\sum_{k=1}^m
    \log\bigg(\frac{\snr_k}{\inr_k}\bigg)^+\nonumber
\\& =       -m + ( m + 1 ) \log( m + 1 )+\sum_{k=1}^m
    \log\bigg(\frac{\snr_k}{\inr_k}\bigg)^+.
    \label{e:manyToOneErrorTerm}
\end{align}
Next, the first and last sum in equation (\ref{Gaussian_sum_rate}) can be simplified as
\begin{align}
&\sum_{k=1}^{m-1}   \log  \left(
    \inr_k+\max\left(\frac{\inr_{k+1}}    {\snr_{k+1}},1\right)\right)
    -\sum_{k=1}^{m-1} \log\left(
    \frac{\inr_{k+1}}{\snr_{k+1}}\right)^+ \nonumber
\\&\leq
     (m-1)+\sum_{k=1}^{m-1} \max\left(\log(\inr_k),\log\left(\frac{\inr_{k+1}}    {\snr_{k+1}}\right)^+\right) -\log\left(
    \frac{\inr_{k+1}}{\snr_{k+1}}\right)^+\nonumber
\\&\leq
     (m-1) + \sum_{k=1}^{m-1} \left(\log(\inr_k)- \log\left(
    \frac{\inr_{k+1}}{\snr_{k+1}}\right)^+\right)^+.\nonumber
\end{align}
Lastly, we upper bound the second-from-last sum in equation (\ref{Gaussian_sum_rate}) as
\begin{align}
  \log\left[2\snr_0+\sum_{i=1}^{m-1}
    \max\left(\frac{\inr_{i+1}}{\snr_{i+1}},1\right)\right] \nonumber
\leq
    \log(m+1)+ \log(\snr_0),
\end{align}
resulting in a cruder, yet simpler, sum-rate bound:
\begin{align*}
r_0+ r_1 +\dots +r_m \nonumber
&\leq
    \sum_{k=1}^m
    \log\bigg(\frac{\snr_k}{\inr_k}\bigg)^+
+\sum_{k=1}^{m-1} \left(\log(\inr_k)- \log\left(
    \frac{\inr_{k+1}}{\snr_{k+1}}\right)^+\right)^+ \nonumber
\\ &\quad
    +\max\left(\log(\inr_m),\log(\snr_0)\right)
 +( m + 2 ) \log( m + 1 )+1.
\end{align*}

\bibliographystyle{ieeetr}
\bibliography{ManyToOneBIB}

\end{document}